%% file: main.tex
\documentclass[conference]{IEEEtran}
\IEEEoverridecommandlockouts
\pdfoutput=1
\usepackage{cite}
\usepackage{amsmath,amssymb,amsfonts}
\usepackage{algorithmic}
\usepackage{graphicx}
\usepackage{textcomp}
\usepackage{xcolor}
\def\BibTeX{{\rm B\kern-.05em{\sc i\kern-.025em b}\kern-.08em
    T\kern-.1667em\lower.7ex\hbox{E}\kern-.125emX}}

\usepackage{amsthm}
\usepackage{mathtools}
\usepackage{stmaryrd}
\usepackage{textcomp}
\usepackage{tikz}
\usetikzlibrary{arrows,automata,positioning,shapes,patterns,arrows.meta,calc}
\usepackage{pifont}
\usepackage{float}
\usepackage[capitalise]{cleveref}
\usepackage{xspace}
\usepackage{ltl}
\usepackage{listings, multicol, multirow}
\usepackage{caption}
\usepackage{wrapfig}
\usepackage{array}
\usepackage{makecell}
\usepackage{booktabs, tabularx}
	\newcolumntype{Y}{>{\centering\arraybackslash}X}

\include{macros}

\newtheorem{theorem}{Theorem}[section]
\newtheorem{lemma}[theorem]{Lemma}
\newtheorem{proposition}[theorem]{Proposition}

\theoremstyle{definition}
\newtheorem{definition}{Definition}[section]

\pgfdeclarelayer{background}
\pgfsetlayers{background,main}

\begin{document}

\title{Smart Contract Synthesis Modulo Hyperproperties
\thanks{This work was partially supported by the German Research Foundation (DFG) as part of the Collaborative Research Center “Foundations of Perspicuous Software Systems” (TRR 248, 389792660), and by the European Research Council (ERC) Grant OSARES (No. 683300). N. Coenen, J. Hofmann, and J. Tillman carried out this work as members of the Saarbr\"ucken Graduate School of Computer Science.
}
}

\author{
\IEEEauthorblockN{Norine Coenen}
\IEEEauthorblockA{\textit{CISPA Helmholtz Center for Information Security}\\
Saarbr\"ucken, Germany \\
norine.coenen@cispa.de}
\and
\IEEEauthorblockN{Bernd Finkbeiner}
\IEEEauthorblockA{\textit{CISPA Helmholtz Center for Information Security}\\
Saarbr\"ucken, Germany \\
finkbeiner@cispa.de}
\and
\IEEEauthorblockN{Jana Hofmann}
\IEEEauthorblockA{\textit{CISPA Helmholtz Center for Information Security}\\
Saarbr\"ucken, Germany \\
jana.hofmann@cispa.de}
\and
\IEEEauthorblockN{Julia Tillman}
\IEEEauthorblockA{\textit{CISPA Helmholtz Center for Information Security}\\
Saarbr\"ucken, Germany \\
julia.tillman@cispa.de}
}

\maketitle
\IEEEpeerreviewmaketitle

\begin{abstract}
	Smart contracts are small but highly security-critical programs that implement wallets, token systems, auctions, crowd funding systems, elections, and other multi-party transactions on the blockchain. A broad range of methods has been developed to ensure that a smart contract is functionally correct. However, smart contracts often additionally need to satisfy certain hyperproperties, such as symmetry, determinism, or an information flow policy. In this paper, we show how a synthesis method for smart contracts can ensure that the contract satisfies its desired hyperproperties. We build on top of a recently developed synthesis approach from specifications in the temporal logic TSL. We present HyperTSL, an extension of TSL for the specification of hyperproperties of infinite-state software. As a preprocessing step, we show how to detect if a hyperproperty has an equivalent formulation as a (simpler) trace property. Finally, we describe how to refine a synthesized contract to adhere to its HyperTSL specification.
\end{abstract}

\begin{IEEEkeywords}
hyperproperties, smart contracts, reactive synthesis, temporal logics
\end{IEEEkeywords}

\section{Introduction}
Smart contracts are small programs which describe transactions between multiple parties. Common smart contracts implement wallets, token systems, auctions, crowd funding systems, and elections. Smart contracts run on the blockchain and thus remove the need for a trusted third party.
Bugs in smart contracts have led to huge monetary losses in the last decade, making the need for rigorous formal foundations for smart contracts obvious.
In the absence of a trusted third party, it is especially important to ensure that a contract implements the agreed order of transactions.
The large body of existing work on the correctness of a contract's control flow mainly consists of verification methods to ensure that a contract exhibits a correct behavior with respect to \emph{trace properties}.


Smart contracts often additionally need to satisfy \emph{hyperproperties}, however.
Hyperproperties describe the relation between several execution traces of the system.
Prominent hyperproperties are information flow policies like noninterference and observational determinism.
Existing work on hyperproperties in smart contracts focuses on verifying concrete information flow policies such as integrity, which are enforced using language-based methods~\cite{DBLP:conf/sp/CecchettiYNM21, cecchetti2020securing}.
Hyperproperties are not limited to information flow policies, though.
To establish the users' trust in a contract, one has to additionally show that a contract satisfies hyperproperties such as robustness and symmetry.
In a voting contract, for example, all candidates should be treated symmetrically, and a vote from the owner of the contract should not count more than any other vote.

In this work, we develop a general-purpose temporal logic, HyperTSL, to express the wide range of hyperproperties relevant in the context of smart contracts.
We then show how to automatically construct the underlying state machine of a smart contract directly from the logical specification of the hyperproperties.
We approach the task from two angles.
First, we investigate the synthesis problem of the logic and show how to approximate the in general undecidable problem to obtain correct-by-design implementations.
Second, we propose a two-step approach to include \htslm properties in a recently proposed smart contract synthesis approach~\cite{scsynt}. \\

\myparagraph{HyperTSL}
We develop two variants of the logic, \htslm and \htslp, which are both based on temporal stream logic (TSL)~\cite{FinkbeinerKPS19}.
TSL extends linear temporal logic (LTL) with the concept of cells and uninterpreted functions and predicates. These mechanisms separate the temporal control flow from the concrete data in the system, which enables reasoning about infinite-state systems. TSL has been applied successfully to synthesize hardware~\cite{DBLP:journals/corr/abs-2101-07232}, functional programs~\cite{DBLP:conf/haskell/Finkbeiner0PS19}, and also to describe the underlying state machines of smart contracts~\cite{scsynt}. As an example, we can state that the winner of an election remains the same after the method \method{close} has been called.
$$
	\G (\method{close} \rightarrow \X \G \update{\field{winner}}{\field{winner}})
$$
The contract's method \method{close} is modeled as a boolean and the field \field{winner} is modeled as a cell.
We define \htslm as a conservative extension of TSL. It adds quantification over multiple executions; predicates, functions, and updates then refer to one of the quantified executions. In an election with two candidates \texttt{A} and \texttt{B}, we can state that two traces have the same winner \texttt{A} if they agree on the votes for \texttt{A}.
\begin{align*}
	\forall \pi \forall \pi' \ldot & \G(\method{voteA}_\pi \leftrightarrow \method{voteA}_{\pi'}) \\
	& \rightarrow \G(\updatepi{\field{winner}}{\const{A}}{\pi} \leftrightarrow \updatepi{\field{winner}}{\const{A}}{\pi'})
\end{align*}
The second hyperlogic \htslp allows to relate executions within a predicate. Here, we can state that the winners on two executions are always the same as long as the votes are always the same.
\begin{align*}
	\forall \pi \forall \pi' \ldot & \G (\texttt{vote}_\pi = \texttt{vote}_{\pi'}) \rightarrow \G (\field{winner}_\pi = \field{winner}_{\pi'})
\end{align*}
Note that the $=$ predicate ranges over two different execution variables $\pi$ and $\pi'$.
In TSL, all functions and predicates are uninterpreted, but $=$ would probably be implemented as actual equality.\\

\myparagraph{Smart contract synthesis from \htslm}
Synthesis constructs correct-by-design systems from a formal specification.
It therefore shifts the development from software writing to specification writing.
This idea fits especially well to smart contracts. Most contracts belong to one of the major classes like token systems or election protocols, and their behavior can be characterized with similar specifications. 
This also extends to hyperproperties: the set of actually relevant hyperproperties is often the same across different systems.

We therefore investigate the synthesis problem of \htslm. Inherited from TSL, the problem is undecidable already for a single quantifier. We show, however, that the $\forall^*$ fragment of \htslm can be approximated in $\forall^*$ HyperLTL, for which there exists a bounded synthesis approach implemented in HyperBoSy~\cite{DBLP:journals/acta/FinkbeinerHLST20}.
For the $\exists^*$ fragment of \htslm, we present an approximation based on LTL satisfiability checking.

As a step towards actual smart contract synthesis, we build on a recent workflow that implements the synthesis of a smart contract control flow from TSL specifications~\cite{scsynt}.
We present a two-step approach to synthesize a contract that adheres to HyperTSL specifications.
First, we check if the combination of a given TSL specification with a $\forall^*$ \htslm specification lets the hyperproperty ``collapse'' to a simple trace property. We show that the check is undecidable in general, but can to some extend be approximated in a decidable fragment of with HyperLTL satisfiability checking.
If this is not the case, we synthesize the most general solution (the winning region) for the TSL specification using the tool presented in~\cite{scsynt}. We then prune the system to find a strategy that implements a $\forall^*$ \htslm property.
Our implementation of the approach shows that we can automatically construct a voting contract which satisfies several hyperproperties.\\

\myparagraph{Contributions}
We present the first solution to automatically construct smart contracts with respect to hyperproperties. 
\begin{itemize}
	\item We define \htslm and \htslp, two hyperlogics based on TSL, and show how they can specify hyperproperties occurring in the context of smart contracts.
	\item We show that the synthesis problem of $\forall^*$ \htslm can be over-approximated by $\forall^*$ HyperLTL synthesis and that the synthesis problem of $\exists^*$ \htslm can be under-approximated by LTL satisfiability checking.
	\item We describe a preprocessing step to detect if a \htslm formula has an equivalent TSL formulation.
	\item We describe and implement a repair-like pruning algorithm that refines a smart contract synthesized from a TSL specification to satisfy a $\forall^*$ \htslm formula. \\
\end{itemize}

\myparagraph{Structure of the paper}
We first give the necessary preliminaries, including the definitions of HyperLTL and TSL. In \Cref{sec:hypertsl}, we present \htslm and \htslp, and show how to specify hyperproperties in a voting contract. We proceed by investigating the realizability problem of \htslm and argue why \htslm is more suited for synthesis than \htslp. In \Cref{sec:SCSynthesis}, we recap the synthesis workflow based on TSL~\cite{scsynt} and discuss the combination of the TSL specification of the voting contract with the \htslm specifications. \Cref{sec:pseudo_hyperprops} discusses the detection of pseudo hyperproperties. \Cref{sec:repair} presents the pruning of a contract with respect to \htslm, including a prototype implementation.

\section{Preliminaries}
We recap the definitions of HyperLTL and TSL.
\subsection{HyperLTL}
Let $\ap$ be an infinite set of atomic propositions. A trace $t$ is an element of $(2^\ap)^\omega$. We write $t[i]$ for the $i$th element of $t$. A trace property $P$ is a set of traces, a hyperproperty $H$ is a set of sets of traces.
HyperLTL~\cite{HyperLTL} describes hyperproperties by extending LTL with explicit trace quantification. It is defined with respect to an infinite set of trace variables $\pathvars$.
\begin{align*}
	\varphi &{}\Coloneqq \forall\pi\ldot\varphi \mid \exists\pi\ldot\varphi \mid \psi \enspace \\
	\psi &{}\Coloneqq a_\pi \mid \neg\psi \mid \psi\land\psi \mid \X\psi \mid \psi\U\psi,
\end{align*}
where $a \in AP$ and $\pi \in \pathvars$.
In HyperLTL, atomic propositions are indexed with the trace variable they refer to.
The semantics of HyperLTL is defined with respect to a set of traces $T$.
Let $\traceAssign : \pathvars \to T$ be a trace assignment that maps trace variables to traces in $T$. 
To update $\traceAssign$, we write $\traceAssign[\pi \mapsto t]$, where $\pi$ maps to $t$ and all other trace variables are as in $\traceAssign$.
The satisfaction relation of HyperLTL is defined with respect to a trace assignment $\traceAssign$, a set of traces $T$, and a point in time $i$.
\begin{alignat*}{3}
	&\traceAssign,T, i \models a_\pi       &&\text{ iff } &&a \in \traceAssign(\pi)[i] \\
	&\traceAssign,T, i \models \neg \psi              && \text{ iff } &&\traceAssign,T, i \not\models \psi \\
	&\traceAssign,T, i \models \psi_1 \land \psi_2         && \text{ iff } &&\traceAssign,T, i \models \psi_1 \text{ and } \traceAssign,T, i \models \psi_2 \\
	&\traceAssign,T, i \models \X \psi                && \text{ iff } &&\traceAssign,T, i+1 \models \psi \\
	&\traceAssign,T, i \models \psi_1 \U \psi_2             && \text{ iff } &&\exists j \geq i \ldot ~ \traceAssign,T,  j \models \psi_2 \\
	& && && \qquad \text{and } \forall i \leq k < j \ldot ~\traceAssign,T, k \models \psi_1 \\
	&\traceAssign,T, i \models \exists \pi \ldot \varphi && \text{ iff } && \exists t \in T \ldot \traceAssign[\pi \mapsto t], T, i \models \varphi\\
	&\traceAssign,T, i \models \forall \pi \ldot \varphi && \text{ iff } && \forall t \in T \ldot \traceAssign[\pi \mapsto t], T, i \models \varphi
\end{alignat*}
As for LTL, we can derive the usual boolean logical connectives and temporal operators ``eventually'' $\F \psi = \top \U \psi$, ``globally'' $\G \psi = \neg \F \neg \psi$, and ``weak until'' $\psi_1 \W \psi_2 = (\psi_1 \U \psi_2) \lor \G \psi_1$.
When specifying smart contracts, we assume specifications to belong to the \emph{syntactic safety} fragment, i.e., that they only contain $\X$, $\G$, and $\W$ as temporal operators.
Formulas from the syntactic safety fragment can be equivalently expressed using \emph{past-time temporal operators}, which reason about the past instead of the future of a trace. ``Yesterday'' $\Y$ is the past-time equivalent of $\X$, and ``since'' $\Since$ corresponds to $\U$. 
We say that a trace set $T$ satisfies a HyperLTL formula $\varphi$, written as $T \models \varphi$, if $\emptyset, T, 0\models \varphi$, where $\emptyset$ denotes the empty trace assignment.

We define the HyperLTL realizability problem as in~\cite{conf/cav/FinkbeinerHLST18}.
We assume $\ap$ to be a disjoint union of a set of input propositions and a set of output propositions, i.e., $\ap = \apin \cupdot \apout$.
A \emph{strategy} $\sigma : (\pow{\apin})^+ \rightarrow \pow{\apout}$  maps finite sequences of input valuations to an output valuation.
For an infinite word $w = w_0 w_1 w_2 \cdots \in (\pow{\apin})^\omega$, the trace generated by a strategy $\sigma$ is defined as $\sigma_w = (w_0 \cup \sigma(w_0))(w_1 \cup \sigma(w_0 w_1))\ldots \in (2^{\ap})^\omega$.
We write $\mathit{traces}(\sigma)$ for the set $\set{\sigma_w \mid w \in (\pow{\apin})^\omega}$.
We say that a strategy $\sigma$ satisfies a HyperLTL formula $\varphi$ over $\ap = \apin \cupdot \apout$ if $\mathit{traces}(\sigma) \models \varphi$.
\begin{definition}[HyperLTL realizability]
	A HyperLTL formula $\varphi$ over $\ap = \apin \cupdot \apout$ is realizable if there is a strategy $\sigma \colon (\pow{\apin})^+ \rightarrow \pow{\apout}$ that satisfies $\varphi$.
\end{definition}

\subsection{TSL}
We introduce the syntax and semantics of TSL following~\cite{scsynt, FinkbeinerKPS19}.
TSL extends LTL with the concept of cells, which hold data from an arbitrary domain.
To abstract from concrete data points, TSL employs uninterpreted functions and predicates.
A TSL formula describes a system that receives a stream of inputs, abstracts from the concrete values using predicates, and produces a stream of cell updates using function applications.
We formally define values, functions, and predicates.
The set of all values is denoted by~$\values$, the Boolean values by $\mathbb{B} \subseteq \values$.
An $n$-ary function $f: \values^n \rightarrow \values$ computes a value from~$n$ values.
An $n$-ary predicate $p: \values^n \rightarrow \mathbb{B}$ assigns a boolean value to $n$ values.
The sets of all functions and predicates are denoted by $\functions$ and $\predicates \subseteq \functions$, respectively.
Constants are both 0-ary functions and values.

Let $\inputs$ and $\cells$ be the set of inputs and cells, and let $\funcSymbols$ and $\predSymbols \subseteq \funcSymbols$ be the set of function and predicate symbols.
\emph{Function terms~$\funcTerm$} are recursively defined by
$$
\funcTerm \Coloneqq \texttt{s}~ \mid ~\texttt{f}~ \tau^f \ldots \, \tau^f
$$
where $\texttt{s} \in \inputs \cup \cells$ and $\texttt{f}$ is an n-ary function symbol.
\emph{Predicate terms~$\predTerm$} are obtained by applying a predicate to function terms.
The sets of all function and predicate terms are denoted by $\funcTerms$ and $\predTerms \subseteq \funcTerms$, respectively.
TSL formulas are built according to the following grammar:
\[ 
\psi \Coloneqq \neg \psi \,\mid\, \psi \land \psi \,\mid\, \X \psi \,\mid\, \psi \U \psi \,\mid\, \predTerm \,\mid\, \update{\cell{c}}{\funcTerm}
\]
where $\cell{c} \in \cells$, $\predTerm \in \predTerms$, and $\funcTerm \in \funcTerms$.
An \emph{update term} $\update{\cell{c}}{\tau_f}$ denotes that the value of function term $\tau_f$ is assigned to cell~$\cell{c}$. We denote the set of update terms with $\updateTerms$.

Function and predicate terms are syntactic objects, i.e. the term $\texttt{p}(\texttt{f}~\texttt{c})$ only becomes meaningful when we assign, for example, \texttt{c} to the value $3$, \texttt{f} to the function $f(x) = x + 1$ and \texttt{p} to $p(x) = x > 0$.
To assign a semantic interpretation to function and predicate symbols, we use an \emph{interpretation} $\assign: \funcSymbols \rightarrow \functions$.
A cell assignment $\cells \rightarrow \funcTerms$ is a total function assigning exactly one function term to each cell.
The set of all assignments $\cells \rightarrow \funcTerms$ is denoted by $\cellAssignments$.
A \emph{computation} $\computation \in \compStream$ describes the control flow of the cells, i.e., which function term is associated with a cell at which point in time.
For every cell $\cell{c} \in \cells$, let $\initial{\cell{c}}$ be a designated value assigned initially to $\cell{c}$.
Input streams~$\inputStream$ are infinite sequences of assignments of inputs to values.
Given a computation, an input stream, and a point in time $\timePoint$, we can \emph{evaluate} a function term.
The evaluation function $\eval: \compStream \times \inputStream \times \mathbb{N} \times \funcTerms \rightarrow \values$ is defined as
\begin{align*}
	&\eval(\computation,\iota,\timePoint,\cell{s}) \coloneqq 
	\begin{cases}
		\iota[\timePoint]~\cell{s} & \cell{s} \in \inputs \\
		\initial{\cell{s}} & \cell{s} \in \cells \land \timePoint = 0 \\
		\eval(\computation,\iota,\timePoint\!-\!1,\computation[\timePoint\!-\!1]~\cell{s}) &  \cell{s} \in \cells \land \timePoint>0
	\end{cases} \\
	& \eval(\computation,\iota,\timePoint,\texttt{f}~ \tau_0 \dots \tau_{n}) \coloneqq \\
	& \phantom{\eval(\computation,\iota,\timePoint,\cell{s}) \coloneqq} ~ \langle \texttt{f} \rangle ~ \eval(\computation,\iota,\timePoint, \tau_0) ~\dots ~\eval(\computation,\iota,\timePoint,\tau_{n})
\end{align*}
Note that $\iota[\timePoint]~\cell{s}$ denotes the value that input stream $\iota$ assigns to input $\cell{s}$ at position $\timePoint$.
Likewise, $\computation[\timePoint]~ \cell{s}$ is the function term that $\computation$ assigns to cell $\cell{s}$ at point in time $\timePoint$.
As an example, to compute the value of cell \cell{x} in step $\timePoint$, we might obtain from the computation that \cell{x} is updated to $\texttt{f}~\cell{x}$ in step $\timePoint$, so we recursively evaluate \cell{x} in step $\timePoint-1$ and apply the function assigned to $\texttt{f}$ to the result.
We evaluate a TSL formula $\psi$ with respect to an assignment function~$\assign$, an input stream $\iota \in \inputStream$, a computation $\computation \in \compStream$, and a time step $\timePoint\in\mathbb{N}$.
\begin{alignat*}{3}
	&\computation, \iota, \timePoint \TSLSatisfaction{\assign} \neg \psi ~ &\quad \text{iff} \quad & \computation, \iota, \timePoint \not\TSLSatisfaction{\assign} \psi\\
	&\computation, \iota, \timePoint \TSLSatisfaction{\assign} \psi_1 \land \psi_2 ~ &\quad \text{iff} \quad & \computation, \iota, \timePoint \TSLSatisfaction{\assign} \psi_1 \text{ and } \computation, \iota, \timePoint \TSLSatisfaction{\assign} \psi_2\\
	&\computation, \iota, \timePoint \TSLSatisfaction{\assign} \X \psi ~ &\quad \text{iff} \quad & \computation, \iota,  \timePoint+1 \TSLSatisfaction{\assign} \psi\\
	&\computation, \iota, \timePoint \TSLSatisfaction{\assign} \psi_1 \U \psi_2 ~ &\quad \text{iff} \quad & ~ \exists j \geq \timePoint \ldot \computation, \iota, j \TSLSatisfaction{\assign} \psi_2 \\ & & & \text{ and }\forall \timePoint \leq k < j \ldot \computation, \iota, k \TSLSatisfaction{\assign} \psi_1 \\
	&\computation, \iota, \timePoint \TSLSatisfaction{\assign} \update{\cell{c}}{\tau} ~ &\quad \text{iff} \quad & ~ \computation[\timePoint]~\cell{c} \equiv \tau\\
	&\computation, \iota, \timePoint \TSLSatisfaction{\assign} \texttt{p}~ \tau_1 \dots \tau_{n} ~ &\quad \text{iff} \quad & ~ \eval(\computation,\iota,\timePoint,\texttt{p}~ \tau_1 \dots \tau_{n})
\end{alignat*}
We use $\equiv$ to syntactically compare two terms. 
An \emph{execution} $(\computation, \iota)$ satisfies a TSL formula~$\psi$ if $\computation, \iota, 0 \TSLSatisfaction{\assign} \psi$ holds. In this case, we write $\computation, \iota \TSLSatisfaction{\assign} \psi$.
A TSL formula $\psi$ is \emph{satisfiable}  iff there exists an interpretation $\assign$ and an execution $(\computation, \iota)$ such that $\computation, \iota \TSLSatisfaction{\assign} \psi$~\cite{DBLP:journals/corr/abs-2104-14988}.

\myparagraph{TSL realizability}
The realizability problem of a TSL formula $\psi$ asks whether there exists a strategy which reacts to predicate evaluations with cell updates according to $\psi$. Formally, a strategy is a function $\sigma : (2^{\predTerms})^+ \rightarrow \cellAssignments$.
Given a strategy $\sigma$ and an input stream $\iota$, we can define the resulting computation $\sigma(\iota)$. To compute the cell updates at point in time $\timePoint$, we require the current inputs as well as the history of cell updates:
\begin{align*}
	\sigma(\iota)[\timePoint] \coloneqq \sigma(\, & \set{\predTerm \in \predTerms \mid \eval(\sigma(\iota), \iota, 0, \predTerm)} \cdot \ldots \\
	& \ldots \cdot \set{\predTerm \in \predTerms \mid \eval(\sigma(\iota), \iota, \timePoint, \predTerm)} \,)
\end{align*}
Note that in order to define $\sigma(\iota)[\timePoint]$, the definition uses $\sigma(\iota)$. This is well-defined since the evaluation function $\eval(\computation, \iota, \timePoint, \tau)$ only uses $\computation[0] \ldots \computation[\timePoint-1]$.

\begin{definition}[TSL realizability \cite{FinkbeinerKPS19}]
	A TSL formula $\psi$ is \emph{realizable} iff there exists a strategy $\sigma : (2^{\predTerms})^+ \rightarrow \cellAssignments$ such that for every input stream $\iota \in \inputStream$ and every assignment function $\assign: \predSymbols \rightarrow \functions$, it holds that $\sigma(\iota), \iota \TSLSatisfaction{\assign} \psi$.
\end{definition}

\section{HyperTSL}
\label{sec:hypertsl}
We extend TSL with execution quantifiers to allow reasoning about hyperproperties for infinite-state systems. There is not a single obvious way to lift TSL to a hyperlogic, so we motivate in the following our choices for \htslm.

\myparagraph{Function and predicate interpretations}
There are two options to deal with function and predicate interpretations. The first is to interpret functions and predicates differently across multiple executions. This would make sense, for example, to model multiple components of a system. Most of the hyperproperties, however, concern the behavior of the same system when executed several times. Functions and predicates like equality and addition mostly stay the same across several executions. We therefore decide to quantify the function and predicate interpretation \emph{before} quantifying executions.

\myparagraph{Domain of updates and predicates}
Another question is whether predicates, functions, and update terms may contain terms evaluated on different executions. In the case of update terms the answer should probably be no: an update of the form $\update{\texttt{c}_\pi}{\texttt{f}(\texttt{c}_{\pi'})}$ would mean that the system could access the value of a cell at the same point in time but on a different execution. This does not seem realistic for a real system, which would have to store all possible values a cell \texttt{c} could hold at some point in time. Additionally, the execution chosen for $\pi'$ differs with every quantifier evaluation.
Predicates ranging over multiple executions, on the other hand, would be useful. One could state a stronger version of observational determinism, namely ``if two traces agree on the value of input \texttt{x}, then also cell \cell{c} is always updated to the same term'', written as $\G (\cell{x}_\pi = \cell{x}_{\pi'}) \rightarrow \G (\updatepi{\cell{c}}{\texttt{f}(\cell{x})}{\pi} \leftrightarrow \updatepi{\cell{c}}{\texttt{f}(\cell{x})}{\pi'})$. We therefore define two logics: \htslm, in which we allow predicates to range only over the \emph{same} execution, and \htslp, which relates multiple executions using predicates.


\subsection{\htslm}
We define \htslm as an extension of TSL. The definition of function terms $\funcTerm$ and predicate terms $\predTerm$ are the same as in TSL. The syntax of \htslm is given as follows.
\begin{align*}
	\varphi &\Coloneqq \forall\pi\ldot\varphi \mid \exists\pi\ldot\varphi \mid \psi \enspace \\
	\psi &\Coloneqq \neg \psi \,\mid\, \psi \land \psi \,\mid\, \X \psi \,\mid\, \psi \U \psi \,\mid\, \predTerm_\pi \,\mid\, \updatepi{\cell{c}}{\funcTerm}{\pi}
\end{align*}
where $\pi \in \pathvars$, $\cell{c} \in \cells$, $\predTerm \in \predTerms$, and $\funcTerm \in \funcTerms$.
We evaluate a \htslm formula $\varphi$ with respect to a set of executions $\exSet \subseteq \compStream \times \inputStream$, an execution assignment function $\traceAssign$, and a point in time $\timePoint$. We omit the cases for $\neg$, $\land$, $\X$, and $\U$, which follow closely the respective cases in TSL.
\begin{alignat*}{3}%
	&\traceAssign, \exSet, \timePoint\TSLSatisfaction{\assign} \updatepi{\cell{c}}{\funcTerm}{\pi} &~ \text{iff} ~ & \projcomp{\traceAssign(\pi)}~i~\cell{c} \equiv \funcTerm\\
	&\traceAssign, \exSet, \timePoint\TSLSatisfaction{\assign} (\texttt{p} \, \tau_1 \dots \tau_{n})_\pi &~ \text{iff} ~ & \eval(\computation,\iota,\timePoint ,\texttt{p} ~ \tau_1 \dots \tau_{n}) \\
	& & & \quad \text{for }\computation = \projcomp{\traceAssign(\pi)} \text{ and }\\
	& & & \quad \phantom{\text{for }}\iota = \projin{\traceAssign(\pi)}  \\
	&\traceAssign, \exSet, \timePoint \TSLSatisfaction{\assign} \exists \pi \ldot \varphi & \text{ iff } & \exists \execution \in \exSet \ldot \traceAssign[\pi \mapsto \execution], \exSet, \timePoint \TSLSatisfaction{\assign} \varphi \\
	&\traceAssign, \exSet, \timePoint \TSLSatisfaction{\assign} \forall \pi \ldot \varphi & \text{ iff } & \forall \execution \in \exSet \ldot \traceAssign[\pi \mapsto \execution], \exSet, \timePoint \TSLSatisfaction{\assign} \varphi	
\end{alignat*}
Here, $\projin{\cdot}$ and $\projcomp{\cdot}$ are used to project an execution on its input stream and its computation. Given an assignment function $\assign: \funcSymbols \rightarrow \functions$, a set of executions $\exSet$ satisfies a \htslm formula $\varphi$, written $E \TSLSatisfaction{\assign} \varphi$, if $\emptyset, \exSet, 0 \TSLSatisfaction{\assign} \varphi$.

\subsection{\htslp}
In \htslp, predicate terms are defined with respect to the set of executions variables $\pathvars$. If $\tau^1, \ldots , \tau^n$ are function terms, $\pi_1, \ldots \pi_n \in \pathvars$, and \texttt{p} is an $n$-ary predicate symbol, then \texttt{p} $\tau^1_{\pi_1} ~ \ldots ~ \tau^n_{\pi_n}$ is a predicate term. The syntax of \htslp is that of \htslm with the exception that we denote predicate terms as $\tau^p_{\pi_1, \ldots, \pi_n}$ instead of $\tau^p_{\pi}$. The semantics of \htslp is that of \htslm except for the semantics of the predicate, which we define as follows.
\begin{align*}
	&\traceAssign, \exSet, \timePoint\TSLSatisfaction{\assign} \texttt{p} ~ \tau^1_{\pi_1} \, \ldots \, \tau^n_{\pi_n} \quad \text{iff} \quad \langle \texttt{p} \rangle ~ v_1 ~ \ldots ~ v_n \\
	&\text{where } v_j = \eval(\projcomp{\traceAssign(\pi_j)},\projin{\traceAssign(\pi_j)},\timePoint ,\tau^j) ~ \text{for } 1 \leq j \leq n 
\end{align*}
Satisfaction of a \htslp formula is defined as for \htslm.

\subsection{HyperTSL for Smart Contracts}
\label{sec:smartContracts}
\label{subsec:votingProtocol}
As for TSL, the strength of \htslm and \htslp is in their use of uninterpreted predicates and functions together with a cell mechanism. That way, the logics can describe control flow properties like ``if predicate $\texttt{p}(\texttt{c})$ evaluates to $\true$, cell \texttt{c} should be updated by applying function \texttt{f} to it''. Since the concrete implementation of \texttt{p} and \texttt{f} is not relevant for the property, the synthesized system can be later completed to a concrete system by implementing \texttt{p} and \texttt{f} as needed.

In this section, we showcase how the logics can express temporal hyperproperties of smart contracts.
As a running example, we use a simple voting contract.
We informally describe the contract and specify hyperproperties that should hold in such a contract. 
This showcases the expressiveness of the logics and makes the differences in expressiveness become clear. We postpone the formal specification of the trace properties of the contract to \Cref{sec:SCSynthesis}. 

For simplicity, we describe the core features of an exemplary voting contract, which are enough to demonstrate typical hyperproperties occurring in a voting scenario. The contract describes an election where users can vote for either candidate \texttt{A} or candidate \texttt{B} by calling methods \method{voteA} and \method{voteB}. The contract can be closed by the owner, who can call method \method{close}. \method{voteA}, \method{voteB}, and \method{close} are modeled as booleans. We assume that in each step, exactly one of \method{voteA}, \method{voteB}, and \method{close} holds.
The contract has two fields, \field{votesA} and \field{votesB}, to store the number of votes recorded for the candidates. Both fields are modeled as cells in TSL. In the smart contract setting, we use the term field and cell interchangeably. Furthermore, \field{winner} holds the current winner chosen by the contract. The winner is updated in every step where the contract receives a vote, so either $\update{\field{winner}}{\const{A}}$ or $\update{\field{winner}}{\const{B}}$ holds globally. \const{A} and \const{B} are constants.

%
The hyperproperties relevant in such a setting are typical hyperproperties such as determinism and symmetry. Consider, for example, the question which candidate is chosen as winner if there is a tie. One would expect a form of symmetry, e.g., that the winner is not simply always candidate \texttt{A}. Similarly, one typically requires some sort of determinism, i.e., that in similar situations the same winner is chosen.
We use the following syntactic sugar.
\begin{align*}
	&\mathit{sameWinner}(\pi, \pi') \coloneqq \\
	& \quad \phantom{\land} (\updatepi{\field{winner}}{\const{A}}{\pi} \leftrightarrow \updatepi{\field{winner}}{\const{A}}{\pi'}) \\
	& \quad \land (\updatepi{\field{winner}}{\const{B}}{\pi} \leftrightarrow \updatepi{\field{winner}}{\const{B}}{\pi'})	
\end{align*}
The formula states that two executions $\pi$ and $\pi'$ choose the same winner.

\myparagraph{Determinism}
We state that the winner of the election is determined by the sequence of votes received.
\begin{align*}%
	\begin{split}	
	&\forall \pi \ldot \forall \pi' \ldot \mathit{sameWinner}(\pi, \pi') \W \\
	& \quad \big[(\method{voteA}_\pi \nleftrightarrow \method{voteA}_{\pi'}) \lor (\method{voteB}_\pi \nleftrightarrow \method{voteB}_{\pi'})\big]
	\end{split}
\end{align*}
The formula states that for any two executions, as long as they receive the exact same votes, they should agree on who is the winner.
Alternatively, we could have stated a local version of determinism using a \texttt{greater} predicate.
Instead of writing \texttt{greater} \field{votesA} \field{votesB}, we use the better readable infix notation with $>$.
Note that the predicate is uninterpreted at this stage altough we intend it to be implemented as expected.
We express that the choice of the winner is strictly determined by the current evaluation of the $>$ predicate on the votes and the last vote.
\begin{align}
	\label{spec:localdeterminism}
	\begin{split}
			&\forall \pi \ldot \forall \pi' \ldot \G \big[ \\
			& \quad ((\field{votesA} > \field{votesB})_\pi \leftrightarrow (\field{votesA} > \field{votesB})_{\pi'} )\\
			& \quad \land ((\field{votesB} > \field{votesA})_\pi \leftrightarrow (\field{votesB} > \field{votesA})_{\pi'} ) \\
			& \quad \land (\method{voteA}_\pi \leftrightarrow \method{voteA}_{\pi'}) \land (\method{voteB}_\pi \leftrightarrow \method{voteB}_{\pi'}) \\%
			& \quad \rightarrow \mathit{sameWinner}(\pi, \pi') \big]
	\end{split}
\end{align}
This formula states that if \texttt{A} and \texttt{B} received the same number of votes, the winner must be the same on both executions.
Note that the formula implicitly entails that any other predicate/input does not influence the winner. For example, if there is a predicate term \texttt{sender = \const{owner}} testing whether the caller of the method is the contract's owner, this should not influence the \field{winner} field.
This formulation of determinism additionally entails that the past of an execution cannot influence the choice of the winner.

In \htslp, we can express determinism by relating executions with predicates. Instead of abstracting from the concrete number of votes using the $>$ predicate, we could state local determinism as follows. As before, the $=$ predicate used in the formula is uninterpreted.
\begin{align*}
	&\forall \pi \ldot \forall \pi' \ldot \\
	& \quad \G \big[ \field{votesA}_\pi = \field{votesA}_{\pi'}\land \field{votesB}_\pi = \field{votesB}_{\pi'} \\
	& \qquad \rightarrow \mathit{sameWinner}(\pi, \pi') \big]
\end{align*}

\myparagraph{Symmetry}
A prominent fairness condition in systems with multiple agents is \emph{symmetry}, which requires that agents are treated symmetrically. In the voting case, this means, e.g., that if two traces swap the votes for \texttt{A} and \texttt{B}, then the winner must also be swapped.
\begin{align}
	\label{spec:symmetry}
	\begin{split}
	&\forall \pi \ldot \forall \pi' \ldot \big[ (\updatepi{\field{winner}}{\const{A}}{\pi} \leftrightarrow \updatepi{\field{winner}}{\const{B}}{\pi'}) \\
	& \qquad \land (\updatepi{\field{winner}}{\const{B}}{\pi} \leftrightarrow \updatepi{\field{winner}}{\const{A}}{\pi'}) \big] \\
	& \quad \W \big[(\method{voteA}_\pi \nleftrightarrow \method{voteB}_{\pi'}) \lor (\method{voteB}_\pi \nleftrightarrow \method{voteA}_{\pi'})\big]
	\end{split}
\end{align}
The fairness property implies that in case of a tie, the winner cannot be brute-forced to always be candidate \texttt{A} or candidate \texttt{B}.

\myparagraph{No harm}
A typical monotonicity criterion in elections states that a vote for the currently leading candidate should not harm the candidate.
Translated to our context this means that a vote for candidate \texttt{A} does not lead to \texttt{B} being the winner instead of \texttt{A}. We formulate this condition as follows in \htslm.
\begin{align}
	\label{spec:noharm}
	\begin{split}
	&\forall \pi \ldot \forall \pi' \ldot \big( \bigwedge_{x \in \predTerms} x_\pi \leftrightarrow x_{\pi'} \big) \, \U \bigg[\X \G \big( \bigwedge_{x \in \predTerms} x_\pi \leftrightarrow x_{\pi'} \big) \\
	& \land \method{voteA}_\pi \land \method{voteB}_{\pi'} \land \big(\bigwedge_{x \in (\predTerms)\backslash \set{\method{voteA}, \method{voteB}}} x_\pi \leftrightarrow x_{\pi'}  \big) \bigg]\\
	& \rightarrow \G (\updatepi{\field{winner}}{\const{A}}{\pi'} \rightarrow \updatepi{\field{winner}}{\const{A}}{\pi})
	\end{split}
\end{align}
The formula describes two executions on which all predicate terms evaluate to the same value except for one point in time at which the first execution receives a vote for \texttt{A} while the second receives a vote for \texttt{B}. Then, \texttt{A} is the winner in the first trace whenever \texttt{A} is in the second one.

\myparagraph{Elections with multiple candidates}
\htslp can specify hyperproperties in contracts with an unknown number of candidates. Assume that instead of \method{voteA} and \method{voteB}, the contract has a single method \method{vote} with a parameter \texttt{cand} indicating which candidate the caller is voting for. We model \method{vote} and \texttt{cand} as inputs. Now, the \field{winner} field might range over a domain of unknown size. We express determinism as follows.
\begin{align*}
	\forall \pi \ldot \forall \pi' \ldot \G \big[ & (\field{winner}_\pi = \field{winner}_{\pi'} )\\
	& \W (\neg (\texttt{cand}_\pi = \texttt{cand}_{\pi'})) \big]
\end{align*}


\section{HyperTSL Synthesis via Reductions}
In this section, we define the synthesis problem of \htslm and argue why \htslm is better suited for synthesis than \htslp.
\htslm inherits the undecidable realizability problem from TSL.
We therefore show that $\forall^*$ \htslm synthesis can be soundly approximated by HyperLTL, i.e, a strategy for the HyperLTL approximation can be translated to a strategy in \htslm.
The $\exists^*$ fragment of \htslm synthesis can be approximated by LTL satisfiability to prove unrealizability.

\begin{definition}
	A \htslm formula $\phi$ is \emph{realizable} iff there exists a strategy $\sigma : (2^{\predTerms})^+ \rightarrow \cellAssignments$ such that for every interpretation $\assign: \predSymbols \rightarrow \functions$ the set constructed from all input streams satisfies $\varphi$, i.e.,  $\set{(\sigma(\iota), \iota) \mid \iota \in \inputStream} \TSLSatisfaction{\assign} \varphi$.
\end{definition}

\begin{figure}[t] 
	\tikzstyle{treenode} = [text width = 3em, inner sep=3pt, text centered]
	\begin{tikzpicture}[
	semithick,
	level/.style={level distance=5cm},
	level 1/.style={sibling distance=3.8cm, level distance=0.8cm},
	level 2/.style={sibling distance=1.9cm, level distance=1.6cm}
	]
	
	\node[treenode] (A) {$\epsilon$}
	child {
		node[treenode] {$\sigma(\emptyset)$}        
		child {
			node[treenode] {$\sigma(\emptyset \emptyset)$ \\ $\vdots$}
			edge from parent
			node[left] {$\emptyset$}
		}
		child {
			node[treenode] {$\sigma(\emptyset \setpx)$ \\ $\vdots$}
			edge from parent
			node[right] {$\setpx$}
		}
		edge from parent 
		node[above left] {$\emptyset$}
	}
	child {
		node[treenode] (B) {$\sigma(\setpx)$}        
		child {
			node[treenode] (C) {$\sigma(\setpx \emptyset)$ \\ $\vdots$}
			edge from parent
			node[left] {$\emptyset$}
		}
		child {
			node[treenode] {$\sigma(\setpx \setpx)$ \\ $\vdots$}
			edge from parent
			node[right] {$\setpx$}
		}
		edge from parent         
		node[above right] {$\setpx$}
	};			
		
	\end{tikzpicture}
	\caption{The strategy tree for a \htslm strategy $\sigma$ branching on a single predicate term $\predTerms = \setpx$.}
	\label{fig:HTSLM-Strategy}
\end{figure}

The above definition defines the realizability problem of \htslm by generalizing the TSL definition, similarly to how HyperLTL generalizes LTL. The strategy has the same type as a TSL strategy, i.e., it generates an execution by reacting to the predicate evaluations for the current input stream. The resulting set of executions must satisfy the \htslm formula. A visualization of the definition can be found in \Cref{fig:HTSLM-Strategy}. 

In \htslp, formulas contain predicates that relate multiple executions. However, a system can still only react to the inputs it receives, i.e., it cannot include the evaluation of such predicates in its decision making. A definition of the realizability problem could therefore only contain predicates not ranging over multiple executions. This makes many formulas unrealizable or only realizable by trivial strategies most likely not intended by the developer. Consider the two \htslp formulas from \Cref{subsec:votingProtocol} again.
In the first formula, the evaluation of the terms $\field{votesA}_\pi = \field{votesA}_{\pi'}$ and $\field{votesB}_\pi = \field{votesB}_{\pi'}$ depends on the chosen executions $\pi$ and $\pi'$ and can therefore not be known by the system. 
The predicate symbol $=$ is uninterpreted and could be implemented with any binary predicate.
A strategy must be winning for all interpretations.
The only winning strategy is therefore the trivial strategy to set \field{winner} to either always \const{A} or always \const{B}.
We observe a similar phenomenon for the second formula. Since $=$ is an uninterpreted predicate, the only winning strategy in this case is to always copy the value of \texttt{cand} to \field{winner}.

The above observations make \htslm the better candidate to extend TSL synthesis to hyperproperties. This is only the case for synthesis, though, and as long as predicates and functions remain uninterpreted. 
If we considered \htslp synthesis with the theory of equality, much more meaningful strategies would be possible. \htslp might also be useful when performing model checking, where formulas can contain predicates that are not part of the system. TSL synthesis with theories has only been studied in one work so far~\cite{DBLP:journals/corr/abs-2108-00090} and is therefore out of the scope of this paper. TSL model checking has not been studied at all so far. Both directions will be exciting future work.

\subsection{Realizability of $\forall^*$ \htslm}
We show that the fact that TSL synthesis can by approximated by LTL synthesis~\cite{FinkbeinerKPS19} carries over to $\forall^*$ \htslm. To do so we recap the approximation of TSL in LTL~\cite{FinkbeinerKPS19}.
The main idea of the approximation is to forget the fact that predicates must evaluate to the same value if evaluated on cells that hold the same value. Then, every predicate $\predTerm$ and update term $\update{\cell{x}}{\funcTerm}$ can be replaced with an atomic proposition $a_{\predTerm}$ and $a_{\cell{x}\_\mathit{to}\_\funcTerm}$, respectively, which is denoted by $\mathit{syntacticConversion}(\psi)$ for a given TSL formula $\psi$. For example, $\mathit{syntacticConversion}(\G (\texttt{p}~\cell{x} \rightarrow \update{\cell{x}}{\texttt{f}~\cell{x}})$ translates to $\G(a_{p\_\cell{x}} \rightarrow a_{\cell{x}\_\mathit{to}\_\texttt{f}\_\cell{x}})$. We call the set of atomic propositions obtained in that translation $\ap_\psi$. Additionally, we ensure that in each step, exactly one update proposition holds for each cell:
\[
	\mathit{cellProps}(\varphi) := \Globally \bigg[ \bigwedge_{\cell{c} \in \cells} ~ \bigvee_{\tau \in \mathcal{T}^{\cell{c}}_{U}} \big( a_\tau \land \bigwedge_{\tau' \in \mathcal{T}^{\cell{c}}_{U} \setminus \{\tau\} } \neg a_{\tau'} \big) \bigg]
\]
where $\updateTerms^\cell{c}$ is the set of all update terms of $\cell{c}$ and $a_\tau$ is the atomic proposition associated with term $\tau$. Now we define
\[
	\transl{\psi} \coloneqq \mathit{syntacticConversion}(\psi) \land \mathit{cellProps}(\psi)
\]
Given an execution $e = (\computation, \iota) \in \compStream \times \inputStream$ and an interpretation $\assign$, we define the corresponding LTL trace $\translT{e}$ over $\ap_\psi$.
\begin{align*}
	\translT{e}[\timePoint] \coloneqq &\set{a_{\predTerm} \mid \eval(\computation, \iota, \timePoint, \predTerm)} \, \cup \\
	 &\set{a_{\cell{x}\_\mathit{to}\_\funcTerm} \mid \computation~\timePoint~\cell{x} \equiv \funcTerm}
\end{align*}

Theorem 1 in~\cite{FinkbeinerKPS19} shows that for a TSL formula $\psi$, if $\transl{\psi}$ is realizable, then $\psi$ is realizable. Actually, their proof shows the following stronger result. 

\begin{proposition}[Proof of Theorem 1 in \cite{FinkbeinerKPS19}]
	\label{prop:tsl_ltl}
	Let $\psi$ be a TSL formula.
	For any execution $e$, any interpretation $\assign$, and any point in time $\timePoint$, $e, \timePoint \TSLSatisfaction{\assign} \psi$ iff $\translT{e}, \timePoint \models \transl{\psi}$.	
\end{proposition}
This results entails that realizability of $\transl{\psi}$ implies realizability of $\psi$; the opposite direction does not hold.
This is because the semantics of TSL forbids \emph{spurious} behavior, which is not captured by the approximation in LTL.
For example, if we have $\update{\cell{x}}{\cell{y}}$ at some point in time, $\texttt{p}(\cell{x})$ will be equivalent to $\texttt{p}(\cell{y})$ but $a_{\texttt{p}\_\cell{x}}$ and $a_{\texttt{p}\_\cell{y}}$ are different atomic propositions with possibly different truth values assigned to them.
Therefore, every execution $e$ can be mapped to a trace $\translT{e}$, but there are traces $t$ over $\ap_\psi$ such that there is no execution $e$ with $\translT{e} = t$.
The approximation is still very valuable, because strategies found for the LTL approximations carry over to strategies for the TSL specification. We show that this result can be lifted to $\forall^*$ \htslm and $\forall^*$ HyperLTL. Given a set of executions $\exSet$, we set $\translT{\exSet} = \set{\translT{\execution} \mid \execution \in \exSet}$ and also lift $\transl{\cdot}$ to \htslm formulas by setting $\transl{\tau_\pi} = (\transl{\tau})_\pi$ with $\tau$ being either a predicate term or an update term.

\begin{lemma}
	\label{lem:equivalence}
	Let $\varphi$ be a \htslm formula. For any set of executions $\exSet$ and any interpretation $\assign$, $\exSet \TSLSatisfaction{\assign} \varphi$ iff $\translT{\exSet} \models \transl{\varphi}$.
\end{lemma}
\begin{proof}
	Let $\varphi = (\exists^*\forall^*)^k \ldot \psi$ be a \htslm formula consisting of $k$ blocks of $\exists^*\forall^*$ quantifiers.
	For the first direction, let $E \TSLSatisfaction{\assign} \varphi$. We show $\translT{\exSet} \models \transl{\varphi}$.
	Let $\varphi[i]$ be the subformula of $\varphi$ starting from the $i$th quantifier block with $1 \leq i \leq k$.
	We keep the invariant to only construct trace assignments $\Pi$ such that there exists an execution assignment $\hat{\Pi}$ with $\Pi(\pi) = \translT{\hat{\Pi}(\pi)}$ and $\hat{\Pi}, \exSet, 0 \TSLSatisfaction{\assign} \varphi[i]$. When choosing the witness traces for the existential variables $\pi^{i+1}_1, \ldots \pi^{i+1}_n$ of the $(i+1)$th quantifier block, we choose $\Pi(\pi^{i+1}_j) = \translT{e^{i+1}_j}$ where $e^{i+1}_j$ is assigned to $\pi^{i+1}_j$ by the proof of satisfaction of $\varphi[i]$ with respect to the current execution assignment $\hat{\Pi}$. Now, by \Cref{prop:tsl_ltl} and a simple induction over the structure of $\psi$, we get $\Pi, \translT{\exSet}, 0 \models \psi$.
	The argument is similar for the other direction since for every $t \in \translT{\exSet}$, there exists an $\execution \in \exSet$ such that $\translT{\execution} = t$.
\end{proof}

Note that the above lemma does not extend to \htslp as predicates over multiple executions cannot be mapped to HyperLTL. From the lemma we get the following theorem. We will also reuse the lemma in \Cref{sec:SCSynthesis}. 

\begin{theorem}
	\label{thm:forallHTSL}
	Let $\varphi$ be a $\forall^*$ \htslm formula. If $\transl{\varphi}$ is realizable, then $\varphi$ is realizable.
\end{theorem}
\begin{proof}
	Let $\ap_\varphi = \apin \cupdot \apout$ be the atomic propositions obtained in the translation $\transl{\varphi}$, where $\apin$ are the propositions generated for predicate terms and $\apout$ the propositions for update terms, respectively.
	Let $\sigma_\text{HyperLTL} : (\pow{\apin})^+ \rightarrow \pow{\apout}$ be the realizing strategy for $\transl{\varphi}$. We claim that the following strategy $\sigma$ realizes $\varphi$.
	\begin{align*}
		&\sigma~ (P_1 \cdot \ldots \cdot P_i) ~\cell{c} \coloneqq \funcTerm \\
		&\text{s.t. } \transl{\update{\cell{c}}{\funcTerm}} \in \sigma_\text{HyperLTL}(\transl{P_1} \cdot \ldots \cdot \transl{P_i})
	\end{align*}
	Here, we lift $\transl{\cdot}$ to sets by $\transl{P_j} = \set{\transl{\predTerm} \mid \predTerm \in P_j}$.
	By definition of $\mathit{cellProps}$, the function term $\funcTerm$ is unique for every $i$ and every cell $\cell{c}$.
	We show that $\sigma$ realizes $\varphi$. 
	Let $\assign$ be any interpretation. Let $\exSet = \set{(\sigma(\iota), \iota) \mid \iota \in \inputStream}$ be the set of executions obtained from $\sigma$ and $T = \mathit{traces}(\sigma)$ be the set of traces obtained from $\sigma_\text{HyperLTL}$.
	Notice that $\translT{E} \subseteq T$, but $T$ may not be a subset of $\translT{E}$. Since $T \models \transl{\varphi}$ and since universal properties are downwards-closed, also $\translT{E} \models \transl{\varphi}$. Therefore, by \Cref{lem:equivalence} we have $\exSet \TSLSatisfaction{\assign} \varphi$.	
\end{proof}

This theorem approximates $\forall^*$ \htslm synthesis by $\forall^*$ HyperLTL synthesis. Even for HyperLTL, this problem is undecidable~\cite{DBLP:journals/acta/FinkbeinerHLST20}. However, there exists an implementation of bounded synthesis for $\forall^*$ HyperLTL in the tool BoSyHyper~\cite{DBLP:journals/acta/FinkbeinerHLST20}. This approach searches for smallest systems implementing the formula. The efficient reduction from $\forall^*$ \htslm to HyperLTL opens the door to apply BoSyHyper to \htslm.

\subsection{Realizability of $\exists^*$ \htslm}

\Cref{thm:forallHTSL} uses the downwards closedness of $\forall^*$ HyperLTL properties. This is necessary since a strategy for \htslm produces ``fewer'' traces than a strategy for the translation of the formula to HyperLTL does. Consequently, the proof does not extend to \htslm formulas with existential quantifiers. We show how to translate the synthesis problem of $\exists^*$ \htslm formulas to a TSL satisfiability problem, which we can approximate by LTL satisfiability checking using the translation from TSL to LTL from~\cite{FinkbeinerKPS19}. In contrast to the case of $\forall^*$ HyperTSL formulas, the approximation of existential properties in LTL can be used to show unrealizability instead of realizability.

\begin{theorem}
	Given an $\exists^*$ \htslm formula $\varphi$, there exists a TSL formula $\psi_\text{TSL}$ such that $\varphi$ is realizable iff $\psi$ is satisfiable.
\end{theorem}
\begin{proof}
	Let $\varphi = \exists \pi_1 \ldots \exists \pi_n \ldot \psi$. We encode the fact that all traces in a model of $\varphi$ are producible by a strategy:
	\begin{align*}
		\varphi_\text{strat} \coloneqq \forall \pi \ldot \forall \pi' \ldot \left[\bigwedge_{\upTerm \in \updateTerms} \upTerm_\pi \leftrightarrow \upTerm_{\pi'}\right] \W \left[\bigvee_{\predTerm \in \predTerms} \predTerm_\pi \nleftrightarrow \predTerm_{\pi'}\right]
	\end{align*}
	The formula states that while two executions have the same predicate evaluations, they have to perform the same updates. Formula $\varphi' = \exists \pi_1 \ldots \exists \pi_n \ldot \psi \land \varphi_\text{strat}$ ensures that the executions chosen as witnesses for $\pi_1, \ldots, \pi_n$ can be arranged in a strategy tree. Thus, $\varphi$ and $\varphi'$ are equi-realizable. The additional conjunct also ensures that $\varphi'$ is satisfiable iff it is realizable. Formula $\varphi'$ is satisfiable iff it is satisfiable by $n$ executions. The $\forall$ quantifiers in $\varphi'$ quantify from these $n$ executions. We can therefore create an equi-satisfiable $\exists^n$ formula by ``unrolling'' the $\forall$ quantifiers. Let $\varphi_\text{strat} = \forall \pi \ldot \forall \pi' \ldot \psi_\text{strat}$. We define:
	\begin{align*}
		\varphi_\text{new} \coloneqq \exists \pi_1 \ldots \exists \pi_n \ldot \psi \land \bigwedge_{1 \leq i, j \leq n} \psi_\text{strat}[\pi \mapsto \pi_i, \pi' \mapsto \pi_j]
	\end{align*}
	Now we have that $\varphi$ is realizable iff $\varphi_\text{new}$ is satisfiable. The above formula is satisfiable iff a TSL formula $\psi_\text{TSL}$ is satisfiable, where we construct $\psi_\text{TSL}$ by creating $n$ copies of each input and each cell. 
\end{proof}
It is an open question whether the two approximations for $\forall^*$ HyperTSL and $\exists^*$ HyperTSL can be combined to approximate formulas with quantifier alternations. The challenge here is that universal properties are best over-approximated to obtain realizability results, whereas existential properties would need an under-approximation.

\section{Hyperproperties for Smart Contract Synthesis}
\label{sec:SCSynthesis}
We have shown that universal and existential \htslm formulas can be approximated by HyperLTL synthesis and LTL satisfiability. For the case of universal properties, this result is, as of now, mostly of theoretic interest. The realizability problem of $\forall^*$ HyperLTL is undecidable~\cite{DBLP:journals/acta/FinkbeinerHLST20}. Bounded hyperproperty synthesis is also still in its infancy.

In the following three sections, we therefore present a more feasible approach to synthesize smart contract control flows.
We exploit the fact that we do not synthesize a system from hyperproperties only, but in combination with trace properties describing the temporal behavior of the contract.
We build on a workflow recently developed for smart contract control flows from TSL specifications~\cite{scsynt} and augment it with methods for HyperTSL.
In this section, we recap the TSL approach by showing how to specify 
some of the trace properties expected of a voting protocol as described in \Cref{subsec:votingProtocol}.
We also show how to include the hyperproperties discussed in \Cref{sec:smartContracts}.
In the next section, we analyze whether the combination of hyperproperties with trace properties leads to ``pseudo'' hyperproperties. 
In \Cref{sec:repair}, we synthesize the most general system from the trace properties and resolve free choices with a repair-like algorithm.

\subsection{Recap: Smart Contract Synthesis from TSL Specifications}

It was recently demonstrated that the temporal behavior of a smart contract can be extensively specified in the syntactic safety fragment of TSL~\cite{scsynt}.
The idea of the approach is to synthesize the underlying state machine of a smart contract from a TSL specification. The state machine is then translated to Solidity code, which can then be augmented with additional functionality to obtain the desired contract.
TSL is a convenient logic to not only describe the correct order of method calls (method ``\method{vote} should not be called after \method{close} has been successfully called'') but also how fields need to be updated to ensure a correct control flow. 
This can be achieved through TSL's cells and uninterpreted functions and predicates. 
The TSL specification abstracts from the concrete implementation of functions and predicates and instead states how fields have to be updated if certain predicates hold (``if method \method{vote} is called, update the \field{winner} field, depending on whether \field{votesA} \texttt{>} \field{votesB}'').

While TSL synthesis is in general undecidable~\cite{FinkbeinerKPS19}, correct solutions can be found through approximation in LTL, which works very well in practice~\cite{FinkbeinerKPS19, DBLP:journals/corr/abs-2101-07232,DBLP:conf/haskell/Finkbeiner0PS19}. By restricting the TSL specification to the past-time fragment, synthesis can be solved comparatively fast using a BDD-based implementation in a tool called \scsynt~\cite{scsynt}.

The synthesis is integrated in a feedback loop as shown in \Cref{fig:TSL_workflow}.
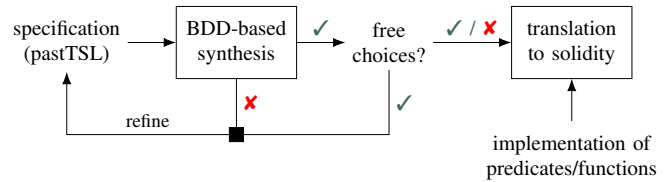
\begin{figure}
		 \scalebox{0.8}{			
					\begin{tikzpicture}[shorten >=0pt,auto]
							
							\node[align=center]		(spec)	at (0,0)		{specification \\ (pastTSL)};
							
							\node[draw,align=center,minimum width=2cm,minimum height=1.2cm]				(synt)	[right=0.8cm of spec]			{BDD-based \\ synthesis};							
							\node[align=center]		(choice)	[right=0.8cm of synt]		{free \\ choices?};							
							\node[draw,align=center,minimum width=2cm,minimum height=1.2cm]		(trans)	[right=1.3cm of choice]		{translation \\ to solidity};							
							\node[draw,fill=black]		(amb)		[below=0.8cm of synt]		{};
							\node[align=center,minimum width=2cm,minimum height=1.2cm]		(impl)	[below=0.7cm of trans]		{implementation of \\ predicates/functions};					
							
							\path[-{Latex[length=2mm]}]	(spec) 	edge		node		{}	(synt)
							(synt)		edge 		node[above] {\color{darkgreen}\ding{51}} (choice)
							(choice)		edge 		node[above] {\color{darkgreen}\ding{51} / \color{red}\ding{56}} (trans)
							(impl) 	edge		node		{}	(trans);
							\draw[-{Latex[length=2mm]}]	(amb) -| node[above,pos=0.25]{\small{refine}}	(spec);
							\draw 
							(synt)		edge		node[right]	{\color{red}\ding{56}}		(amb)
							(choice) |- node[right,pos=0.25]{\color{darkgreen}\ding{51}}	(amb);
							
					\end{tikzpicture}
			 }
	\caption{Workflow for the synthesis of a smart contract control flow as developed in~\cite{scsynt}.}	
	\label{fig:TSL_workflow}
\end{figure}
The user provides a past-time TSL specification, which is translated to LTL. If the LTL formula is unrealizable, this might be either a specification error or is due to the LTL approximation (in most cases it's the former). If the specification is realizable, \scsynt returns the full \emph{winning region} -- a Mealy machine that implements the specification and includes all realizing strategies. If there is more than one admissible strategy, the specification might be incomplete and should be refined. If not, or if either strategy is fine, a strategy is extracted from the winning region. In the last step, the developer provides an implementation of the functions and predicates (such that the implementation respects the assumptions stated in the specification), and \scsynt generates a Solidity smart contract.

\scsynt specializes on a correct temporal control flow and does not check if the provided function and predicate implementations satisfy the stated assumptions. This can be checked by specialized tools based on SMT solving. For recent approaches on how to combine TSL synthesis with reasoning over theories and syntax-guided synthesis, see~\cite{DBLP:journals/corr/abs-2108-00090, tslSygus}.

\subsection{Voting Protocol}
\label{subsec:usecase}
The part of the TSL voting specification that is relevant for this paper is depicted in Specification~\ref{spec:voting}.
\begin{specification}{Functional specification of the voting contract in TSL.}{spec:voting}
	Inputs: (*\method{voteA}*), (*\method{voteB}*), (*\method{close}*), (*\method{reveal}*), sender
	Cells: (*\field{votesA}*), (*\field{votesB}*)
	Functions: addOne
	Predicates: >, =
	Constants: owner(), A(), B()
	
	--- Assumptions ---
	// Mutex on methods
	(*$\Globally$*)((*\method{voteA}*) -> ! (*\method{voteB}*) && ! (*\method{close}*));	
	(*$\Globally$*)((*\method{voteB}*) -> ! (*\method{voteA}*) && ! (*\method{close}*));
	(*$\Globally$*)((*\method{close}*) -> ! (*\method{voteA}*) && ! (*\method{voteB}*));
	(*$\Globally$*)((*\method{voteA}*) || (*\method{voteB}*) || (*\method{close}*));
	
	//Arithmetic assumptions on > predicate
	(*$\Globally$*) !((*\field{votesA}*) > (*\field{votesB}*) && (*\field{votesB}*) > (*\field{votesA}*));
	//Initial assumption on number of votes
	!((*\field{votesA}*) > (*\field{votesB}*)) && !((*\field{votesB}*) > (*\field{votesA}*));
	
	--- Requirements ---
	(*$\Globally$*)((*\method{close}*) -> sender = owner());		
	(*$\Globally$*)((*\method{close}*) -> (*$\Next \Globally$*) !(*\method{voteA}*) && !(*\method{voteB}*));
	
	--- Obligations ---
	(*$\Globally$*)((*\method{voteA}*) -> [(*\field{votesA}*) <- addOne (*\field{votesA}*)]);
	(*$\Globally$*)((*\method{voteB}*) -> [(*\field{votesB}*) <- addOne (*\field{votesB}*)]);
	
	(*$\Globally$*)(! (*\method{voteA}*) -> [(*\field{votesA}*) <- (*\field{votesA}*)]);
	(*$\Globally$*)(! (*\method{voteB}*) -> [(*\field{votesB}*) <- (*\field{votesB}*)]);
	
	(*$\Globally$*)(((*\method{voteA}*) || (*\method{voteB}*)) 
			-> [(*\field{winner}*) <- A()]||[(*\field{winner}*) <- B()]);
	(*$\Globally$*)(((*\method{voteA}*) || (*\method{voteB}*)) && (*\field{votesA}*) > (*\field{votesB}*) 
			-> [(*\field{winner}*) <- A()]);
	(*$\Globally$*)(((*\method{voteA}*) || (*\method{voteB}*)) && (*\field{votesB}*) > (*\field{votesA}*) 
			-> [(*\field{winner}*) <- B()]);			
	
	(*$\Globally$*)((*\method{close}*) -> (*$\Globally$*)[(*\field{winner}*) <- (*\field{winner}*)]);	
\end{specification}%
Specifications are divided into three parts. \emph{Assumptions} describe properties we can assume, e.g., that in each step only one method is called. We also formulate assumptions on the predicates which must be met by the user providing the implementation of the functions and predicates. \emph{Requirements} state in which order and under which conditions methods can be called. \emph{Obligations} describe how the contract's fields must be updated.

The specification has four inputs. There is one input for each method indicating the called method. The \texttt{sender} input contains the caller of the method (\texttt{msg.sender} in Solidity). Cells describe the fields of the contract. We use a function \texttt{addOne}, which will be implemented as an increment. There are two predicates $>$ and $=$, which will be implemented as the actual ``greater'' and ``equals'' predicate. Lastly, we use three constants, one for each candidate \texttt{A} and \texttt{B}, and one for the owner of the contract. Constants will be translated to immutable Solidity fields which are set in the constructor.

Method calls are assumed to be atomic, i.e., exactly one method is called at a time. This follows other work on modeling the control flow of smart contracts, e.g.,~\cite{DBLP:conf/fc/MavridouL18, DBLP:conf/vstte/0001LCPDBNF19}.
Initially, neither candidate has more votes than the other, what can be achieved by setting \field{votesA} and \field{votesB} to 0. Lastly, keeping in mind that the $>$ predicate will be implemented as the actual ``greater'' predicate, we assume that $\field{votesA} > \field{votesB}$ and $\field{votesB} > \field{votesA}$ cannot hold true at the same time.
The requirements ensure that \method{close} can only be invoked by the owner of the contract and that no vote can be cast after the voting is closed.
The obligations describe the updates of the fields. When \method{voteA} or \method{voteB} is invoked, the respective field counting the votes must be increased. Otherwise, the field must self-update, i.e., stay unchanged. If a vote is cast, then the contract must update the winner to either \const{A} or \const{B}. If \field{votesA} holds more votes than \field{votesB}, the \field{winner} field is set to \const{A} and correspondingly to \const{B} if there are more votes in \field{votesB}.
Once \method{close} has been called, the winner cannot change anymore.

In this paper, we include in Specification~\ref{spec:voting} only the trace properties which interact with the hyperproperties we are interested in.
To obtain a satisfactory voting protocol, one would, for example, also record the voters who already cast their vote and require that every person only votes once by using a predicate. This could be achieved by using a field \field{voters} and guard calls to \method{voteA} and \method{voteB} with an uninterpreted predicate \texttt{inVoters \field{voters} sender}, which is later implemented as expected.

The winning region synthesized from the above specification is depicted in \Cref{fig:voting_tsl}. We write \field{w} instead of \field{winner}, and \field{vA} and \field{vB} instead of \field{votesA} and \field{votesB}, respectively. We also use \texttt{+1} instead of \texttt{addOne}. We simplify the system by summarizing multiple transitions between the same states with boolean formulas. For example, $(\neg (\field{vA} > \field{vB}) \land \update{\field{w}}{\const{B}}) \lor (\neg (\field{vB} > \field{vA}) \land \update{\field{w}}{\const{A}})$ expresses that \texttt{A} can be chosen as winner as long as \texttt{B} does not have strictly more votes.
\begin{figure*}[t]
	\centering
	\footnotesize
	\begin{tikzpicture}[initial text=, 
		->,
		node distance=4cm,
		state/.style = {circle, draw, minimum size=5mm,
			inner sep=0pt, outer sep=0pt},
		state/.default = 6pt  
		]

		\node[state, initial] (q1) {$q_1$};
		\node[state, above right = 0.4cm and 7.5cm of q1] (q3) {$q_3$};		
		\node[state, below right = 0.4cm and 7.5cm of q1] (q2) {$q_2$};
		
		\draw	(q1) edge[above,align=left] node[sloped] {\method{close} $\land$ \texttt{sender = \const{owner}} $\land$ \\ $\neg (\field{vA} > \field{vB}) \land \neg (\field{vB} > \field{vA})$\\ $\land \update{\field{vA}}{\field{vA}} \land \update{\field{vB}}{\field{vB}} \land \update{\field{w}}{\field{w}}$} (q3)
		(q1) edge[below,align=left] node[sloped] {$\neg (\field{vA} > \field{vB}) \land \neg (\field{vB} > \field{vA}) \land(\update{\field{w}}{\const{A}} \lor \update{\field{w}}{\const{B}})$ \\ $\land$ \big[(\method{voteA} $\land \update{\field{vA}}{\texttt{+1}~\field{vA}} \land \update{\field{vB}}{\field{vB}}$) \\ \qquad $\lor$ (\method{voteB} $\land \update{\field{vB}}{\texttt{+1}~\field{vB}} \land \update{\field{vA}}{\field{vA}}$)\big]} (q2)
		(q2) edge[loop right,align=left] node[yshift=-3mm] {$\big[(\neg (\field{vA} > \field{vB}) \land \update{\field{w}}{\const{B}}) \lor (\neg (\field{vB} > \field{vA}) \land \update{\field{w}}{\const{A}}) \big]$\\ $\land$ \big[ (\method{voteA} $\land  \update{\field{vA}}{\texttt{+1}~\field{vA}} \land \update{\field{vB}}{\field{vB}}$) \\ \qquad $\lor$ (\method{voteB} $\land \update{\field{vB}}{\texttt{+1}~\field{vB}} \land \update{\field{vA}}{\field{vA}}$) \big]} (q2)
		(q2) edge[right,align=left] node {\method{close} $\land$ \texttt{sender = \const{owner}} $\land$ $\update{\field{w}}{\field{w}}$} (q3)
		(q3) edge[loop right,align=left] node {\method{close} $\land$ \texttt{sender = \const{owner}} $\land$ $\update{\field{w}}{\field{w}}$}(q3);
	\end{tikzpicture}
	\caption{Winning region for Specification~\ref{spec:voting} as synthesized by \scsynt~\cite{scsynt}.} 
	\label{fig:voting_tsl}
\end{figure*}
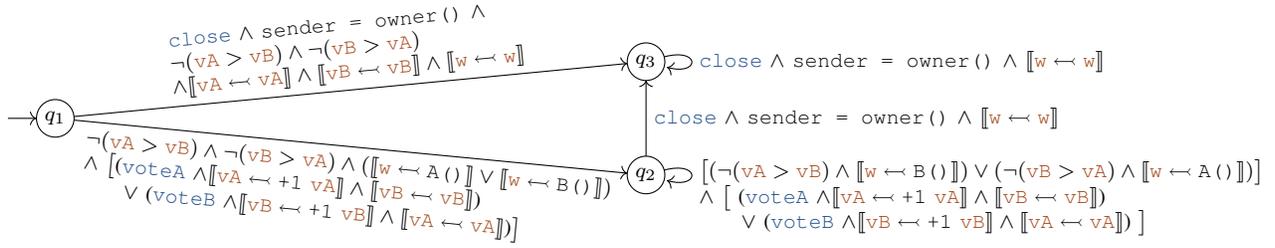

From the winning region we can see that the TSL specification leaves the free choice which candidate is the winner when neither one has the majority of votes.
Strategies permitted by the winning region are, for example, to always choose \texttt{A} as winner, or go for the candidate who got the last vote.
Which strategies are considered to be suitable is described by the \htslm formulations of local determinism (Formula~\ref{spec:localdeterminism} in \Cref{sec:smartContracts}), symmetry (Formula~\ref{spec:symmetry} in \Cref{sec:smartContracts}), and the no harm property (Formula~\ref{spec:noharm} in \Cref{sec:smartContracts}).
Our goal is to combine the TSL specification with the \htslm specifications to obtain a satisfactory strategy.
Since the specification of the voting contract abstracts from the concrete number of votes cast with the \texttt{>} predicate, we need to add the following assumption to make the combination of the specifications realizable.
\begin{align*}
	&\forall \pi \ldot \forall \pi' \ldot \\
	& ~\big[ ((\field{votesA} > \field{votesB})_\pi \leftrightarrow (\field{votesA} > \field{votesB})_{\pi'}) \\
	& \quad \land ((\field{votesB} > \field{votesA})_\pi \leftrightarrow (\field{votesB} > \field{votesA})_{\pi'})\big] \\
	& ~ \W (\method{voteA}_\pi \nleftrightarrow \method{voteA}_{\pi'} \lor \method{voteB}_\pi \nleftrightarrow \method{voteB}_{\pi'}) 
\end{align*}
The property states that the evaluation of the $>$ predicate is the same for every two executions as long as they receive the same sequence of votes.
The property is satisfied when implementing \texttt{addOne} as increment and $>$ as ``greater than''.

When combining the TSL specification of the voting contract with local determinism, symmetry,  and the no harm property, we observe that there is only one valid strategy left to resolve a tie.
Local determinism states that the winner may only depend on the greater predicate on the votes and the vote cast in the current step. This forbids to take into account the past of the trace, e.g., by choosing the winner differently in the first and in the second step.
Symmetry excludes to brute force and always let the same candidate win in case of a tie. Together with local determinism, the winner must therefore truly depend on the current vote.
In combination with local determinism and symmetry, the no harm property forbids to let \texttt{A} win if the current vote is for \texttt{B} (and vice versa).
This leaves as only option to resolve a tie to let the candidate win who got the current vote.
In this case, we could therefore replace the hyperproperties with a trace property describing exactly this strategy, which leads to a much easier synthesis problem.


%

%

\section{Pseudo Hyperproperties}
\label{sec:pseudo_hyperprops}


The previous section has shown that it can be hard to see the implications of combining several hyperproperties with trace properties. It might be the case that the conjunction leaves only one strategy to resolve free choices.
To recognize such situations is not easy for a human.
We therefore propose to preprocess the specification by checking whether the hyperproperties in conjunction with the trace properties effectively define a trace property.
To do so, we introduce the notion of pseudo hyperproperties.
We first investigate the problem for the general definition based on equivalence checking, and subsequently consider the special case of synthesis.

\begin{definition}
	A hyperproperty $H$ is a pseudo hyperproperty iff there is a trace property $P$ such that
	$$
		\forall T \subseteq (\pow{\ap})^\omega \ldot T \in H \text{ iff } \forall t \in T \ldot t \in P
	$$
\end{definition}
If $H$ is a pseudo hyperproperty, $H$ describes the trace property $P = \bigcup_{T \in H} T$.
For proofs we will often use the following fact.
\begin{proposition}
	If $H$ is a pseudo hyperproperty, then $H$ is closed under union and subsets.
\end{proposition}
The downwards closure property in the above proposition implies that only hyperproperties expressible with $\forall^*$ formulas can be pseudo hyperproperties.
The following proposition establishes the convenient fact that if a $\forall^*$ \htslm formula $\varphi$ is a pseudo hyperproperty, then the formula obtained by using only a single execution variable in the body of the formula is equivalent to $\varphi$.
The proposition is close to the corresponding result for HyperLTL, which has been observed in context of using HyperLTL synthesis for the synthesis of linear distributed architectures~\cite{DBLP:journals/acta/FinkbeinerHLST20}.
Given a \htslm formula $\varphi = \forall \pi_1 \ldots \forall \pi_n \ldot \psi$, we define its $\forall^1$ counterpart as $\varphi[\pi] \coloneqq \forall \pi \ldot \psi[\pi_1 \mapsto \pi, \ldots ,\pi_n \mapsto \pi]$.

\begin{proposition}
	A $\forall^*$ \htslm formula $\varphi$ describes a pseudo hyperproperty iff $\varphi \leftrightarrow \varphi[\pi]$.
\end{proposition}
\begin{proof}
	We observe that a hyperproperty $H$ is a pseudo hyperproperty iff
	$$
		\forall T \subseteq (\pow{\ap})^\omega \ldot T \in H \text{ iff } \forall t \in T \ldot \set{t} \in H
	$$
	Thus, for any set of executions $\exSet$ and any interpretation $\assign$, we have $\exSet \TSLSatisfaction{\assign} \varphi$ iff $\forall \execution \in \exSet \ldot \set{\execution} \TSLSatisfaction{\assign} \varphi$ iff $\exSet \TSLSatisfaction{\assign} \varphi[\pi]$.
\end{proof}
The proposition describes how to check if a \htslm property is a pseudo hyperproperty, namely by checking if $\varphi$ is equivalent to $\varphi[\pi]$. If it is, the TSL formula $\varphi[\pi]$ can be used for synthesis. The check is undecidable for \htslm, however.
\begin{theorem}
	It is undecidable whether a \htslm formula $\varphi$ describes a pseudo hyperproperty.
\end{theorem}
\begin{proof}
	The proof follows from undecidability of the satisfiability problem of TSL~\cite{DBLP:journals/corr/abs-2104-14988}. The formula $\varphi = \forall \pi \ldot \forall \pi'. (\texttt{p}~a)_\pi \leftrightarrow (\texttt{p}~a)_{\pi'}$ is not a pseudo hyperproperty as there exists an interpretation of \texttt{p} such that $\exSet$ consists of executions for which $\texttt{p}~a$ always evaluates to $\mathit{true}$, and $\exSet'$ consists of executions for which $\texttt{p}~a$ always evaluates to $\mathit{false}$. $\exSet$ and $\exSet'$ are not closed under union for $\varphi$. Now, let a TSL formula $\psi$ be given for which we assume w.l.o.g. that \texttt{p} is a fresh predicate and $a$ a fresh input. Since $\forall \pi \ldot \forall \pi' \ldot \mathit{false}$ is a pseudo hyperproperty, the formula $\forall \pi \ldot \forall \pi' \ldot ((\texttt{p}~a)_\pi \leftrightarrow (\texttt{p}~a)_{\pi'}) \land \psi_\pi$ is a pseudo hyperproperty iff $\psi$ is unsatisfiable. Here, $\psi_\pi$ denotes the formula obtained by lifting $\psi$ to a \htslm formula by annotating predicate and update terms with the trace variables $\pi$.
\end{proof}

%

Fortunately, the above negative result comes with the remedy that we can approximate the problem in HyperLTL, which is a result of \Cref{lem:equivalence}.

\begin{theorem}
	Let $\varphi$ be a \htslm formula. If $\transl{\varphi}$ is a pseudo hyperproperty, then so is $\varphi$.
\end{theorem}
\begin{proof}
	Closedness under union and subsets follows from \Cref{lem:equivalence} and since for every set of executions $\exSet$ and interpretation $\assign$, $\translT{\exSet} = \bigcup_{\execution \in \exSet} \translT{\execution}$, i.e., $\translT{\exSet \cup \exSet'} = \translT{\exSet} \cup \translT{\exSet'}$, and $\exSet \subseteq \exSet'$ iff $\translT{\exSet} \subseteq \translT{\exSet'}$.
\end{proof}
Using the above result, we can translate a $\forall^*$ \htslm formula to HyperLTL and check its equivalence with the corresponding $\forall^1$ formula. For HyperLTL, this check is decidable, because of the decidability of the $\exists^*\forall^*$ fragment of HyperLTL~\cite{DBLP:journals/acta/FinkbeinerHLST20}: a HyperLTL formula $\varphi = \forall \pi_1 \ldots \forall \pi_n \ldot \psi$ is equivalent to its $\forall^1$ counterpart $\varphi[\pi]$ iff $\varphi[\pi] \rightarrow \varphi$ (because of the semantics of the $\forall$). By merging the quantifiers, this implication is valid iff the formula 
$$
	\exists \pi_1 \ldots \exists \pi_n \ldot \forall \pi \ldot \neg \psi \land \psi[\pi_1 \mapsto \pi, \ldots \pi_n \mapsto \pi]
$$
is unsatisfiable. Thus, if the check returns UNSAT, the \htslm formula is guaranteed to describe a pseudo hyperproperty, so we can replace the hyperproperty with its TSL version.

So far, we described the check for pseudo hyperproperties in terms of satisfiability checking. This can be useful to detect superfluous specifications or mistakes. For example, the developer could have already specified in TSL that in case of a tie, the current vote determines the winner.
\begin{align*}
	&\G (\neg (\field{votesA} > \field{votesB}) \land \neg (\field{votesB} > \field{votesA}) \\
	& \quad \rightarrow \update{\field{winner}}{\const{A}} \leftrightarrow \method{voteA})
\end{align*}
With this specification, the hyperproperties stated in \Cref{sec:smartContracts} are entailed by the specification. In the case of the ``no harm'' property, for example, the check with EAHyper reveals correctly within $0.003$ seconds that the conjunction of the trace properties with the ``no harm'' property is a pseudo hyperproperty.

Without the above TSL specification however, the check for the conjunction of local determinism, symmetry, and the no harm property is labeled as a hyperproperty by EAHyper, even though there is only one possible strategy. The reason for that is the difference between realizability and satisfiability. Consider the following two traces.
\begin{align*}
	& (1)~ \set{\method{voteA}, \update{\field{winner}}{\const{A}}}^\omega \\
	& (2)~ \set{\method{voteB}, \update{\field{winner}}{\const{A}}}^\omega
\end{align*}
In the first trace, neither $\field{votesA} > \field{votesB}$ nor $\field{votesB} > \field{votesA}$, and there is always a vote for \texttt{A}, who is always the winner. The second trace is similar but there is always a vote for \texttt{B}. As single sets, both traces satisfy all three hyperproperties. Together, they do not (because of the symmetry requirement). Therefore, EAHyper returns that the conjunction of the properties is not a pseudo hyperproperty (even in combination with the full contract specification). 
When solving the realizability problem, however, we do not need to find an \emph{equivalent} trace property, it is enough to find a trace property which evaluates the same on all sets generated by strategies.
We therefore adapt the notion of pseudo hyperproperties to realizability.
We give the definition for hyperproperties of traces over $\apin \cupdot \apout$, but the definition applies to traces of any type, particularly also to hyperproperties over TSL-like executions.


\begin{definition}
	Let $H$ be a hyperproperty over atomic propositions $\ap = \apin \cupdot \apout$.
	$H$ is a pseudo hyperproperty in the context of realizability if there is a trace property $P$ such that
	\begin{align*}
		&\forall \sigma: (2^{\apin})^+ \rightarrow 2^{\apout}\ldot \\
		&\qquad \mathit{traces}(\sigma) \in H \text{ iff } \forall t \in \mathit{traces}(\sigma) \ldot t \in P
	\end{align*}
\end{definition}
The definition implies that a HyperTSL formula $\varphi$ is a pseudo hyperproperty in the context of realizability iff there is a TSL formula $\psi$ such that a strategy $\sigma$ realizes $\varphi$ iff $\sigma$ realizes $\psi$.
In the context of realizability, however, $\psi$ may not be equivalent to $\varphi[\pi]$. To show why, we give an example in HyperLTL but the reasoning carries over to HyperTSL. Consider the example of the following HyperLTL formulas over input $i$ and output $o$, which are abstract, simplified versions of local determinism, symmetry, and the no harm property.
\begin{align*}	
	&\forall \pi \ldot \forall \pi' \ldot (i_\pi \leftrightarrow i_{\pi'}) \rightarrow (o_\pi \leftrightarrow o_{\pi'}) \\ 
	&\forall \pi \ldot \forall \pi' \ldot (i_\pi \nleftrightarrow i_{\pi'}) \rightarrow (o_\pi \nleftrightarrow o_{\pi'}) \\
	&\forall \pi \ldot \forall \pi' \ldot i_\pi \land \neg i_{\pi'} \land o_{\pi'} \rightarrow o_{\pi}
\end{align*}
The properties are not temporal, so they are realizable iff there are values $o_1, o_2$ such that $\set{\set{i, o_1}^\omega, \set{\neg i, o_2}^\omega}$ satisfies the three formulas. Let $\varphi$ be their conjunction.
There are only four possible assignments of $o_1, o_2$ to boolean values, each of which corresponds to a possible strategy. Indeed, for every possible strategy, $\varphi$ is satisfied iff the trace property $i \leftrightarrow o$ is satisfied. However, the same does not hold for the trace property $\true$, which is equivalent to $\varphi[\pi]$.
Unfortunately, the general problem is undecidable already for HyperLTL, due to the undecidability of the realizability problem of the $\forall^*$ fragment of HyperLTL.

\begin{theorem}
	Given a HyperLTL formula $\varphi$, it is in general undecidable if there exists an LTL formula $\psi$ such that for all strategies $\sigma$, $\mathit{traces}(\sigma) \models \varphi$ iff $\forall t \in \mathit{traces}(\sigma) \ldot t \models \psi$.
\end{theorem}
\begin{proof}
	We show undecidability by a reduction from the $\forall^*$ HyperLTL realizability problem, which is undecidable~\cite{DBLP:journals/acta/FinkbeinerHLST20}. Let $\rho$ be a $\forall^*$ HyperLTL formula over $\ap = \apin \cupdot \apout$. 
	We assume w.l.o.g. that $\apin$ is non-empty.
	We define $\varphi$ as follows, where $o$ is an output proposition that does not occur in $\rho$.
	$$
	\varphi \coloneqq \rho \land \forall \pi \ldot \forall \pi' \ldot o_\pi \leftrightarrow o_{\pi'}
	$$	
	We claim that $\rho$ is unrealizable iff there exists an LTL formula $\psi$ such that for all strategies $\sigma$, $\mathit{traces}(\sigma) \models \varphi$ iff $\forall t \in \mathit{traces}(\sigma) \ldot t \models \psi$.
	If $\rho$ is unrealizable, then we choose $\psi := \false$. Since $\rho$ is unrealizable, it holds for all $\sigma$ that $\mathit{traces}(\sigma) \not\models \varphi$ iff $\forall t \in \mathit{traces}(\sigma) \ldot t \models \psi$.
	For the other direction, assume that $\rho$ is realizable and that there exists a suitable LTL formula $\psi$. Let $\sigma$ be the realizing strategy of $\rho$.
	Since $o$ is fresh for $\rho$, we can extend $\sigma$ to $\sigma_1$ and $\sigma_2$, where $\sigma_1$ 
	adds $o$ to the first output (for any input), and $\sigma_2$ does not add $o$ (also for any input).
	Both strategies realize $\varphi$, therefore, by our assumption, both strategies realize $\psi$.
	Now, let $i \in \apin$ be any input. We define $\sigma_3$ as the strategy that adds $o$ to the first output exactly if $i$ holds in the first position of the input sequence.
	Since every trace generated by $\sigma_3$ is either a trace of $\sigma_1$ or $\sigma_2$, $\sigma_3$ realizes $\psi$. However, it does not realize $\varphi$, which contradicts our assumption.
\end{proof}

Lastly, we observe that we can decide if a HyperTSL formula is a pseudo hyperproperty if the formula contains as conjunct a local determinism formula like the one in our running example. We define local determinism in the general case as follows.
\begin{align*}
	&\mathit{localDeterminism} := \\
	& \qquad \forall \pi \ldot \forall \pi' \ldot \G ((\bigwedge_{\predTerm \in \predTerms} \predTerm_\pi \leftrightarrow \predTerm_{\pi'}) \rightarrow (\bigwedge_{\tau^c \in \updateTerms} \tau^c_\pi\leftrightarrow \tau^c_{\pi'}))
\end{align*}

\begin{proposition}
	For every $\forall^*$ HyperTSL formula $\varphi$ over predicate terms $\predTerms$ and update terms $\updateTerms$, if $\varphi = \mathit{localDeterminism} \land \varphi'$, then deciding whether $\varphi$ is a pseudo hyperproperty in the context of realizability is equivalent to a Boolean SAT problem.
\end{proposition}
The above proposition follows from the observation that $\mathit{localDeterminism} \land\varphi$ is realizable iff there is a positional strategy that always assigns the same cell updates for the same predicate evaluations, independently of the trace's past. In this case, there are only finitely many combinations of predicate evaluations and the problem becomes finite.

\section{Resolving Choices with Repair}
\label{sec:repair}
If a hyperproperty is not a pseudo hyperproperty, we synthesize the winning region of the trace property and then check if free choices can be resolved such that the hyperproperty is satisfied. This is a repair algorithm additionally respecting the distinction between inputs and outputs.
We first give a formal definition of the problem we are interested in for \htslm formulas. As observed several times, the formal problem is undecidable for \htslm, but a sound approximation can be achieved through HyperLTL.
We then discuss how to simplify the problem for $\forall^*$ formulas and present a prototype implementation, which successfully repairs the synthesized voting contract with respect to determinism, symmetry, and the no harm property.

The result of synthesizing the extended control flow of a smart contract from TSL specifications is a Mealy machine as depicted in \Cref{fig:voting_tsl}. A \emph{Mealy machine} is a tuple $(Q, \Sigma, \delta, q_0)$, where $Q$ is a set of states, $\Sigma = 2^\ap$ with $\ap = \apin \cupdot \apout$ is a labeling alphabet, $\delta : Q \times \Sigma \rightarrow Q$ is a transition function, and $q_0$ is the initial state.

\begin{definition}
	A \emph{free choice} in a Mealy machine $(Q, \Sigma, \delta, q_0)$ consists of a state $q \in Q$ and an input $i \in 2^{\apin}$ such that there are at least two outputs $o_1, o_2 \in 2^{\apout}$ with $(q, i \cup o_1, q_1) \in \delta$ and $(q, i \cup o_2, q_2) \in \delta$ for some $q_1$, $q_2$.
\end{definition}

\begin{definition}
	A Mealy machine $\hat{S}$ is a \emph{refinement} of a Mealy machine $S$ if $\hat{Q} = Q$, $\hat{\Sigma} = \Sigma$, $\hat{q_0} = q_0$ and $\hat{\delta} \subseteq \delta$ such that for every $q \in Q, i \in 2^{\apin}$, if there are $o \in 2^{\apout}$ and $q' \in Q$ such that $\delta(q,i\cup o, q')$, then there are $\hat{o} \in 2^{\apout}$ and $\hat{q'} \in \hat{Q}$ such that $\hat{\delta}(q,i \cup \hat{o}, \hat{q'})$ and $\delta(q,i\cup \hat{o}, \hat{q'})$.
\end{definition}

The above definition ensures that if $S$ defines a set of strategies for a TSL property $\psi$, then a refinement $\hat{S}$ still describes at least one strategy for $\psi$. Our goal is to refine $S$ such that $\hat{S}$ models a \htslm property $\varphi$.

The system $S$ is produced from the LTL approximation, it might therefore contain spurious traces which, for any interpretation, cannot be produced by a combination of an input stream and a computation. We define when such a system satisfies a \htslm formula.
\begin{definition}
	Let $\varphi$ be a \htslm formula over $\predTerms, \updateTerms$ and let $S$ be a Mealy machine over $\ap = \transl{\predTerms} \cup \transl{\updateTerms}$. \emph{$S$ models $\varphi$} if for every interpretation $\assign$:
	\begin{align*}
		\set{(\computation, \iota) \mid \exists t \in \mathit{traces}(S) \ldot t = \translT{(\computation, \iota)}} \TSLSatisfaction{\assign} \varphi
	\end{align*}	
\end{definition}
The definition creates for each interpretation the set of executions which have a trace through $S$ and then checks if $\varphi$ is satisfied on these sets. We show that we cannot check the satisfaction of a \htslm formula on such a system directly.
\begin{lemma}
	Given a \htslm formula $\varphi$, it is undecidable whether a Mealy machine $S$ over $\ap = \transl{\predTerms} \cup \transl{\updateTerms}$ models $\varphi$.
\end{lemma}
\begin{proof}
	We show that the problem is already undecidable for TSL formulas expressed in \htslm as $\forall \pi \ldot \psi_\pi$, where $\psi_\pi$ is the $\pi$-indexed version of a TSL formula $\psi$ .
	We proceed by reduction from the unsatisfiability problem of TSL. Let $\psi$ be a TSL property. $\psi$ is unsatisfiable if for every $\assign$ and every $\computation$, $\iota$, it holds that $\computation, \iota \TSLSatisfaction{\assign} \neg \psi$.
	Let $S$ be the system that produces every trace over $\ap = \transl{\predTerms} \cup \transl{\updateTerms}$. Then, $S$ models $\forall \pi \ldot \neg \psi_\pi$ iff $\psi$ is unsatisfiable.
\end{proof}

However, owing to \Cref{lem:equivalence}, for every $\forall^*$ \htslm formula $\varphi$, if a Mealy machine $S$ models $\transl{\varphi}$, then it also models $\varphi$.

\begin{theorem}
	Let a $\forall^*$ \htslm formula and a Mealy machine $S$ over $\ap = \transl{\predTerms} \cup \transl{\updateTerms}$ be given. If $S$ models $\transl{\varphi}$, then $S$ models $\varphi$.
\end{theorem}
\begin{proof}
	For every $\assign$, let $S_\assign = \set{(\computation, \iota) \mid \exists t \in \mathit{traces}(S) \ldot t = \translT{(\computation, \iota)}}$. By construction, $\translT{S_\assign} \subseteq \mathit{traces}(S)$. Since $S \models \transl{\varphi}$ and since universal properties are downwards closed, $\translT{S_\assign} \models \transl{\varphi}$ and by \Cref{lem:equivalence} we have $S_\assign \TSLSatisfaction{\assign} \varphi$.
\end{proof}

Based on this theorem, we repair the synthesized transition system with respect to the corresponding HyperLTL formula. 
To do so, we enumerate all refinements which leave exactly one transition for each choice. Thereby, the refinement still implements a strategy. Since the HyperLTL is a universal formula, this approach detects a correct refinement iff there is one.
Note that also with respect to HyperLTL, this cannot be a complete method, as $\forall^*$ HyperLTL synthesis is undecidable~\cite{DBLP:journals/acta/FinkbeinerHLST20}.
The general problem of repairing HyperLTL formulas has been discussed in~\cite{DBLP:conf/atva/BonakdarpourF19}, the problem is \texttt{NP}-complete for $\forall^*$ HyperLTL.
An approach similar to our has been described for controller synthesis in~\cite{DBLP:conf/csfw/BonakdarpourF20}, which distinguishes between controllable and uncontrollable inputs.
Implementations have not been developed yet. 

\section{Implementation and Experiments}
\label{subsec:implementation}
We implemented a prototype of the repair algorithm as a Python script. Given a synthesized smart contract in form of a Mealy machine and a universal HyperLTL formula, the script checks if the complete system satisfies the property. If not, it searches for free choices, performs a self-composition, and checks if one of the possibilities to resolve the choices satisfies the LTL body of the HyperLTL formula. For LTL model checking we use the state-of-the-art model checker nuXmv~\cite{DBLP:conf/cav/CavadaCDGMMMRT14}.
Using our implementation, we conducted experiments on different version of the voting contract and on a blinded auction contract.

\subsection{Voting Contract}
We repaired four variants of the voting protocol with respect to the hyperproperties described in this paper. The results are depicted in \Cref{tab:implementation}.
None of the three systems initially satisfies either of the properties, the number of calls given in the table therefore refers to the number of refinements tested with nuXmv.
\begin{table}
	\caption{Results of the prototype implementation of the repair algorithm. Times (t) are given in seconds. \# refers to the number of nuXmv calls that were needed to find a repair.
	}
	\renewcommand{\arraystretch}{1.0}
	\setlength{\tabcolsep}{5pt}
	\begin{tabularx}{\columnwidth}{@{} l *{8}{c} @{}}
		\toprule
		 &
		\multicolumn{2}{c}{\textbf{only} \method{vote}} & 
		\multicolumn{2}{c}{\textbf{+} \method{close}} & 
		\multicolumn{2}{c}{\textbf{+ owner}} &
		\multicolumn{2}{c}{\textbf{full}}\\
		\textbf{Property} & \textbf{t} & \textbf{\#} & \textbf{t} & \textbf{\#} & \textbf{t} & \textbf{\#} & \textbf{t} & \textbf{\#} \\
		\midrule
		\makecell[l]{Local determinism} & 0.170 & 1 & 0.475 & 1 & 2.577 & 1 & 7.049 & 1 \\
		\midrule
		\makecell[l]{Local symmetry} & 0.229 & 2 & 1.251 & 6 & 64.90 & 86 & 308.1 & 120 \\
		\midrule
		\makecell[l]{Global no harm} & 0.163 & 1 & 0.473 & 1 & 2.786 & 1 & 219.9 & 86 \\
		\midrule
		\makecell[l]{Determinism} & 0.147 & 1 & 0.612 & 1 & 27.63 & 35 & 6.825 & 1 \\
		\midrule		
		\makecell[l]{Symmetry} & 0.254 & 2 & 1.630 & 6 & 29.03 & 35 & 6.571 & 1 \\
		\midrule
		\makecell[l]{No harm} & 0.170 & 1 & 0.704 & 1 & 2.993 & 1 & 90.81 & 35 \\
		\midrule
		\makecell[l]{Determinism \\ + Symmetry \\ + No harm} & 0.274 & 3 & 2.105 & 6 & 217.7 & 256 & 760.4 & 256 \\
		\bottomrule
	\end{tabularx}	
	\label{tab:implementation}
\end{table}

The first variant (``only \method{vote}'') reduces the contract to the core specification needed for the hyperproperties to make sense. It has only a \method{vote} method and can be implemented as a single-state system. It has two state-input combinations with a free choice, in each case there are two possible transitions.
The second variant (``+ \method{close}'') adds the \method{close} method and initial assumptions, resulting in a three-state system with four choices, again each with two options.
The third variant (``+ owner'') is the contract as described in Specification~\ref{spec:voting}. It has eight choices. This results in $2^8$ different combinations that might need to be checked.

Lastly, we extended the voting contract with additional features following the voting example of the Solidity Documentation~\cite{voting}. 
The extended version additionally records the registered voters in a field \field{voters} and which addresses have voted in a field \field{voted}.
The contract also has two additional methods: \method{giveRightToVote} may be called by the owner of the contract and adds addresses to \field{voters}; \method{getWinner} may be called after the voting has been closed to learn which candidate won.
The synthesized state machine has again eight choices but is naturally larger than the state machine depicted in \Cref{fig:voting_tsl}.

Determinism, local determinism, symmetry, and the no harm property indicate the formulas given in \Cref{sec:smartContracts}. For symmetry, we included the necessary assumption described in \Cref{subsec:usecase}.
The ``global no harm'' property formulates the no harm property with a single $\G$ instead of $\W$.
Local symmetry states symmetry with respect to the greater predicate as well as the \method{voteA} and \method{voteB} inputs, similar to local determinism.

\subsection{Blind Auction}
As a second contract, we specified a blind auction in TSL, following~\cite{auction}. Similar to the voting protocol, we had to restrict ourselves to a finite number of bidders, otherwise the hyperproperties would not have been expressible in \htslm.
The idea of a blind auction is that bidders send a hash of their actual bid and deposit a value which might be higher or lower than the actual bid.
After the bidding is closed, bidders reveal their bids. If the hash fits the revealed bid and the deposited value is higher than the actual bid, the bid is valid. The winner of the auction is the bidder with the highest valid bid.
Bidders can also withdraw their deposits once a higher bid was revealed.

We specified the contract with two bidders. We use methods \method{bidA}, \method{bidB}, \method{closeBidding}, \method{revealA}, \method{revealB}, \method{closeRevealing}, and \method{withdraw}.
One of the temporal requirements is that bids can only be placed as long as the bidding has not been closed.
$$
	\G (\method{closeBidding} \rightarrow \X \G \neg (\method{bidA} \lor \method{bidB}))
$$
The specification also reasons about fields \field{bidsA}, \field{bidsB}, \field{highestBidder}, \field{highestBid}. One of the obligations on the fields is to update \field{highestBid} if bidder \texttt{A} reveals a correct bid which is higher than all bids revealed so far.
\begin{align*}
	\G( & \method{revealA}\, \land\, \texttt{valid}~ \texttt{bid}~ \texttt{secret} \\
	& \land \texttt{bid} > \field{highestBid}\\
	& \rightarrow \update{\field{highestBid}}{\texttt{bid}})
\end{align*}
In the above specification, \texttt{bid} and \texttt{secret} are parameters of the method \method{revealA}.
As for the voting specification, the TSL specification describes the control flow of the contract, i.e., in which order methods have to be called and how fields need to be updated. Other functionality like checking if a bid is valid needs to be implemented in the predicate \texttt{validBid} to obtain a working contract.

\begin{wraptable}{r}{5.5cm}
	\caption{Repairing the blind auction state machine. Times (t) are given in seconds. \# refers to the number of nuXmv calls.
	}
	\renewcommand{\arraystretch}{1.0}
	\setlength{\tabcolsep}{7pt}
	\centering
	\begin{tabularx}{5cm}{@{} l *{2}{c} @{}}
		\toprule
		\textbf{Property} & \textbf{t} & \textbf{\#} \\
		\midrule
		\makecell[l]{Local determinism} & 11.13s & 1  \\
		\midrule
		\makecell[l]{Local symmetry} & 17.42s & 2 \\
		\bottomrule
	\end{tabularx}	
	\label{tab:implementation2}
\end{wraptable}
Similar to the voting case, we left the free choice which bid is stored as highest bid if the currently revealed bid is the same is the current highest bid. The synthesized state machine leaves free choices at two nodes.
Using our implementation, we obtained strategies to resolve a tie in a way that local determinism or local symmetry are satisfied.
We report the running times and number of calls to nuXmv in \Cref{tab:implementation2}.

\subsection{Evaluation}

Our experiments show that the hyperproperties discussed in this paper are realizable on top of variants of a voting contract and a blinded auction. The evaluation shows that the runtime mainly depends on the number of choice options that had to be tested with \textsc{nuXmv}. The runtime also increases with larger formulas and larger state machines.

This implementation is only a prototype to see if the idea to resolve free choices via repair works for relevant hyperproperties. There are many options to improve performance. For example, one could switch to a model checker that is based on B\"uchi or Parity automata (e.g., spot~\cite{duret.16.atva2}). Like this, the LTL formula (which stays the same for every call) would have to be translated to an automaton only once.

\section{Related Work}
Our work brings together research on formal methods for smart contract control flows, reactive synthesis using TSL, and logics for hyperproperties.

\myparagraph{Formal methods for smart contract control flows}
Most closely related work is the recent approach to synthesize the control and data flow of a smart contract from TSL specifications~\cite{scsynt}. We build on the fact that they synthesize the full winning region of the specification, which we refine to find a strategy that satisfies a \htslm specification.
Another approach to synthesize smart contract control flows is~\cite{DBLP:journals/corr/abs-1906-02906}, which synthesizes the order of method calls from an LTL specification.
Formal guarantees have been also obtained by modeling the control flow of a contract directly with a finite state machine and transform the state machine to Solidity code. In the \textsc{FSolidM} framework~\cite{DBLP:conf/fc/MavridouL18} the state machine can be modeled in a graphical editor.
A lot of work has been invested to verify a contract against correctness properties instead of generating it automatically.
The \textsc{VeriSol}~\cite{DBLP:conf/vstte/0001LCPDBNF19, DBLP:journals/corr/abs-1812-08829} framework checks a smart contract against properties defined with finite state machines.
\textsc{VerX}~\cite{DBLP:conf/sp/PermenevDTDV20} and \textsc{SmartPulse}~\cite{DBLP:conf/sp/StephensFMLD21} verify that a contract satisfies specifications given in LTL-like logics, where \textsc{SmartPulse} can handle not only safety but also liveness properties.

\myparagraph{Temporal Stream Logic}
TSL is a temporal logic that has been developed to simplify reasoning about infinite-state systems by separating control and data~\cite{FinkbeinerKPS19}. Due to its expressiveness, TSL synthesis is undecidable in general but sound CEGAR loops based on LTL synthesis work well in practice~\cite{FinkbeinerKPS19}. TSL has been successfully applied to specify and synthesize an arcade shooter game running on an FPGA~\cite{DBLP:journals/corr/abs-2101-07232} and functional reactive programs~\cite{DBLP:conf/haskell/Finkbeiner0PS19}.
The TSL synthesis approach has been extended with decidable theories like linear integer arithmetic~\cite{DBLP:journals/corr/abs-2108-00090}. Here, counterexamples provided by the CEGAR loop are checked for consistency with the theory using an SMT solver.
TSL modulo first-order theories has also been studied with respect to its satisfiability problem~\cite{DBLP:journals/corr/abs-2104-14988}.

\myparagraph{Hyperproperties}
The interest in hyperproperties, or relational properties, stems from information flow policies like noninterference and observational determinism. Since the term has been coined for the general class of properties~\cite{DBLP:conf/csfw/ClarksonS08}, a range of logics have been developed for the specification of hyperproperties. Examples are temporal logics like HyperLTL, HyperCTL$^*$~\cite{HyperLTL}, and HyperPDL\cite{DBLP:conf/concur/GutsfeldMO20}, first-order and second-order logics~\cite{DBLP:journals/corr/FinkbeinerZ16, DBLP:conf/lics/CoenenFHH19}, or asynchronous extensions mostly based on HyperLTL~\cite{DBLP:conf/cav/BaumeisterCBFS21, DBLP:conf/lics/BozzelliPS21, DBLP:journals/pacmpl/GutsfeldMO21}.
All of the above logics are targeted to specify properties in the context of hardware and therefore assume a finite state space.
Recently, the temporal logic of actions (TLA) has also been used to verify hyperproperties via self-composition~\cite{DBLP:conf/csfw/LamportS21}.
The synthesis from hyperproperties is known to be challenging. HyperLTL can express distributed architectures, which makes the synthesis from $\forall^*$ properties already undecidable~\cite{DBLP:journals/acta/FinkbeinerHLST20}.
A simpler variant of synthesis is controller synthesis, where the state space is given. Related to our approach, the resulting problem is similar to the repair problem and is therefore decidable for HyperLTL~\cite{DBLP:conf/csfw/BonakdarpourF20}.
The classic program repair problem does not distinguish between inputs and outputs and has been studied for HyperLTL as well~\cite{DBLP:conf/atva/BonakdarpourF19}.
The analysis of smart contract with respect to hyperproperties has not received enough attention yet.
A recent approach uses type-checking to ensure information flow policies such as integrity in smart contracts in the context of reentrancy attacks~\cite{DBLP:conf/sp/CecchettiYNM21}.
Other hyperproperties identified in the context of smart contracts are integrity properties to prevent reentrancy attacks~\cite{DBLP:conf/post/GrishchenkoMS18, DBLP:journals/pacmpl/AlbertGRRRS20}

\section{Conclusion and Future Work}

We have presented two logics, \htslm and \htslp, which extend the temporal logic TSL to express hyperproperties. Inherited from TSL, both logics express the control flow of infinite-state systems by abstracting from concrete data with cells and uninterpreted functions and predicates.
\htslp predicates relate several executions, which makes the logic more expressive than \htslm. \htslm, on the other hand, is more suitable for synthesis tasks.
We have demonstrated that both logics can express typical hyperproperties that occur in the context of smart contracts.
The realizability problem of the universal fragment of \htslm can be approximated with HyperLTL realizability, the unrealizability problem of the existential fragment can be approximated with LTL satisfiability.
As a step towards making synthesis from hyperproperties reality, we have described two approaches to integrate universal \htslm specifications in the synthesis workflow of smart contracts specified with TSL.
The first approach checks for overly complicated specifications which have equivalent formulations in TSL. The second approach refines the winning region obtained from the TSL synthesis to resolve choices in a way that satisfy the hyperproperty.

\myparagraph{Future work}
This work has opened many interesting paths for future work.
To profit from the expressiveness of \htslp, one would have to include theories (e.g., the theory of equality) into the synthesis from \htslp specifications. Prior work has made steps in that direction for TSL~\cite{DBLP:journals/corr/abs-2108-00090, DBLP:journals/corr/abs-2104-14988}, one could extend these ideas to \htslp. Using theories, \htslp would also be interesting in the context of model checking smart contracts or other infinite-state systems.
Another open question is how to approximate the HyperTSL realizability problem for formulas with quantifier alternations.

%
%
%
\bibliographystyle{IEEEtran}
\bibliography{bib}

\end{document}

%% file: macros.tex
\definecolor{pastelorange}{HTML}{F79A3D}
\definecolor{pastelblue}{HTML}{2874ae}
\definecolor{pastelgreen}{HTML}{61B940}
\definecolor{pastelviolet}{HTML}{A370AB}
\definecolor{darkblue}{HTML}{1C4987}
\definecolor{chestnut}{HTML}{A24516}
\definecolor{darkgreen}{HTML}{426A5A}

\newcommand{\myparagraph}[1]{\noindent\textbf{#1.}}

\newcommand{\true}{\mathit{true}}
\newcommand{\false}{\mathit{false}}
\renewcommand{\phi}{\varphi}
\newcommand{\ldot}{\mathpunct{.}}
\newcommand{\set}[1]{\{ #1 \}}
\newcommand{\cupdot}{\mathbin{\dot\cup}}
\newcommand{\pow}[1]{2^{#1}}

\newcommand{\apin}{\ap_\mathit{in}}
\newcommand{\apout}{\ap_\mathit{out}}

\newcommand{\Next}{\LTLnext}
\newcommand{\Globally}{\LTLsquare}

\newcommand{\Yesterday}{\LTLcircleminus}
\newcommand{\Y}{\Yesterday}

\newcommand{\Since}{\LTLsince}

\newcommand{\U}{\LTLuntil}
\newcommand{\W}{\LTLweakuntil \,}
\newcommand{\X}{\LTLnext}
\newcommand{\G}{\LTLglobally}
\newcommand{\F}{\LTLeventually}

\newcommand{\ap}{\mathit{AP}}

\newcommand{\traceAssign}{\Pi}
\newcommand{\pathvars}{\mathbb{V}}

\newcommand{\cell}[1]{{\normalfont{\texttt{#1}}}}
\newcommand{\update}[2]{\llbracket #1 \leftarrowtail #2 \rrbracket}
\newcommand{\values}{\mathcal{V}}
\newcommand{\functions}{\mathcal{F}}
\newcommand{\predicates}{\mathcal{P}}
\newcommand{\predTerm}{\tau^p}
\newcommand{\funcTerm}{\tau^f}
\newcommand{\upTerm}{\tau^u}
\newcommand{\predTerms}{\mathcal{T}_P}
\newcommand{\funcTerms}{\mathcal{T}_F}
\newcommand{\updateTerms}{\mathcal{T}_U}
\newcommand{\assign}{{\langle \hspace{-1pt}\cdot \hspace{-1pt}\rangle}}
\newcommand{\funcSymbols}{\Sigma_F}
\newcommand{\predSymbols}{\Sigma_P}
\newcommand{\inputs}{\mathbb{I}}

\newcommand{\cells}{\mathbb{C}}
\newcommand{\cellAssignments}{\mathcal{C}}
\newcommand{\computation}{\varsigma}
\newcommand{\compStream}{\cellAssignments^\omega}
\newcommand{\initial}[1]{\mathit{init}_{#1}}
\newcommand{\inputStream}{\mathcal{I}^\omega}

\newcommand{\eval}{\eta_{\assign}}
\newcommand{\TSLSatisfaction}[1]{\models_{#1}}
\newcommand{\const}[1]{\texttt{#1()}}
\newcommand{\timePoint}{i}

\newcommand{\htslm}{HyperTSL\@\xspace}
\newcommand{\htslp}{HyperTSL\textsubscript{rel}\@\xspace}
\newcommand{\updatepi}[3]{{\llbracket #1 \leftarrowtail #2 \rrbracket}_{#3}}
\newcommand{\projcomp}[1]{\#_\computation(#1)}
\newcommand{\projin}[1]{\#_\iota(#1)}
\newcommand{\exSet}{E}
\newcommand{\execution}{e}
\newcommand{\transl}[1]{[#1]_{\text{atomic}}}
\newcommand{\translT}[1]{\transl{#1}^\assign}

\newcommand{\setpx}{\set{p(x)}}

\newcommand{\method}[1]{\texttt{\color{darkblue}#1}}
\newcommand{\field}[1]{\texttt{\color{chestnut}#1}}

\newcommand{\scsynt}{\textsc{SCSynt}\@\xspace}

\lstdefinelanguage{Solidity}{
	keywords=[1]{anonymous, assembly, assert, balance, break, call, callcode, case, catch, class, constant, continue, constructor, contract, debugger, default, delegatecall, delete, do, else, emit, event, experimental, export, external, false, finally, for, function, gas, if, implements, import, in, indexed, instanceof, interface, internal, is, length, library, log0, log1, log2, log3, log4, memory, modifier, new, payable, pragma, private, protected, public, pure, push, require, return, returns, revert, selfdestruct, send, solidity, storage, struct, suicide, super, switch, then, this, throw, transfer, true, try, typeof, using, value, view, while, with, addmod, ecrecover, keccak256, mulmod, ripemd160, sha256, sha3}, 
	keywordstyle=[1]\color{darkblue}\bfseries,
	keywords=[2]{address, bool, byte, bytes, bytes1, bytes2, bytes3, bytes4, bytes5, bytes6, bytes7, bytes8, bytes9, bytes10, bytes11, bytes12, bytes13, bytes14, bytes15, bytes16, bytes17, bytes18, bytes19, bytes20, bytes21, bytes22, bytes23, bytes24, bytes25, bytes26, bytes27, bytes28, bytes29, bytes30, bytes31, bytes32, enum, int, int8, int16, int24, int32, int40, int48, int56, int64, int72, int80, int88, int96, int104, int112, int120, int128, int136, int144, int152, int160, int168, int176, int184, int192, int200, int208, int216, int224, int232, int240, int248, int256, mapping, string, uint, uint8, uint16, uint24, uint32, uint40, uint48, uint56, uint64, uint72, uint80, uint88, uint96, uint104, uint112, uint120, uint128, uint136, uint144, uint152, uint160, uint168, uint176, uint184, uint192, uint200, uint208, uint216, uint224, uint232, uint240, uint248, uint256, var, void, ether, finney, szabo, wei, days, hours, minutes, seconds, weeks, years},	
	keywordstyle=[2]\color{teal}\bfseries,
	keywords=[3]{block, blockhash, coinbase, difficulty, gaslimit, number, timestamp, msg, data, gas, sender, sig, value, now, tx, gasprice, origin},	
	keywordstyle=[3]\color{pastelviolet!80!black}\bfseries,
	identifierstyle=\color{black},
	sensitive=false,
	comment=[l]{//},
	morecomment=[s]{/*}{*/},
	commentstyle=\color{gray}\ttfamily,
	stringstyle=\color{red!70!black}\ttfamily,
	morestring=[b]',
	morestring=[b]"
}

\newcounter{specification}
\lstnewenvironment{specification}[2]{
	
	\setcounter{lstlisting}{\value{specification}}
	\lstset{caption={\normalsize\normalfont{#1}},
		label={#2},
		captionpos=b,
		float=t
	}
}{\addtocounter{specification}{1}}

\newcounter{solidity}
\lstnewenvironment{solidity}[2]{
	
	\setcounter{lstlisting}{\value{solidity}}
	\lstset{caption={\normalsize\normalfont{#1}},
		label={#2},
		captionpos=b,
		float=t,
		language=Solidity
	}
}{\addtocounter{solidity}{1}}

\DeclareCaptionFormat{listing}{#1#2#3}

%% file: main.bbl
\begin{thebibliography}{10}
\providecommand{\url}[1]{#1}
\csname url@samestyle\endcsname
\providecommand{\newblock}{\relax}
\providecommand{\bibinfo}[2]{#2}
\providecommand{\BIBentrySTDinterwordspacing}{\spaceskip=0pt\relax}
\providecommand{\BIBentryALTinterwordstretchfactor}{4}
\providecommand{\BIBentryALTinterwordspacing}{\spaceskip=\fontdimen2\font plus
\BIBentryALTinterwordstretchfactor\fontdimen3\font minus
  \fontdimen4\font\relax}
\providecommand{\BIBforeignlanguage}[2]{{%
\expandafter\ifx\csname l@#1\endcsname\relax
\typeout{** WARNING: IEEEtran.bst: No hyphenation pattern has been}%
\typeout{** loaded for the language `#1'. Using the pattern for}%
\typeout{** the default language instead.}%
\else
\language=\csname l@#1\endcsname
\fi
#2}}
\providecommand{\BIBdecl}{\relax}
\BIBdecl

\bibitem{DBLP:conf/sp/CecchettiYNM21}
\BIBentryALTinterwordspacing
E.~Cecchetti, S.~Yao, H.~Ni, and A.~C. Myers, ``Compositional security for
  reentrant applications,'' in \emph{42nd {IEEE} Symposium on Security and
  Privacy, {SP} 2021, San Francisco, CA, USA, 24-27 May 2021}.\hskip 1em plus
  0.5em minus 0.4em\relax {IEEE}, 2021, pp. 1249--1267. [Online]. Available:
  \url{https://doi.org/10.1109/SP40001.2021.00084}
\BIBentrySTDinterwordspacing

\bibitem{cecchetti2020securing}
------, ``Securing smart contracts with information flow,'' in
  \emph{International Symposium on Foundations and Applications of Blockchain},
  2020.

\bibitem{scsynt}
\BIBentryALTinterwordspacing
B.~Finkbeiner, J.~Hofmann, F.~Kohn, and N.~Passing, ``Reactive synthesis of
  smart contract control flows,'' 2022. [Online]. Available:
  \url{https://arxiv.org/abs/2205.06039}
\BIBentrySTDinterwordspacing

\bibitem{FinkbeinerKPS19}
B.~Finkbeiner, F.~Klein, R.~Piskac, and M.~Santolucito, ``{T}emporal {S}tream
  {L}ogic: {S}ynthesis {B}eyond the {B}ools,'' in \emph{Computer Aided
  Verification - 31st International Conference, {CAV} 2019, New York City, NY,
  USA, July 15-18, 2019, Proceedings, Part {I}}, ser. Lecture Notes in Computer
  Science, I.~Dillig and S.~Tasiran, Eds., vol. 11561.\hskip 1em plus 0.5em
  minus 0.4em\relax Springer, 2019, pp. 609--629.

\bibitem{DBLP:journals/corr/abs-2101-07232}
\BIBentryALTinterwordspacing
G.~Geier, P.~Heim, F.~Klein, and B.~Finkbeiner, ``Syntroids: Synthesizing a
  game for fpgas using temporal logic specifications,'' \emph{CoRR}, vol.
  abs/2101.07232, 2021. [Online]. Available:
  \url{https://arxiv.org/abs/2101.07232}
\BIBentrySTDinterwordspacing

\bibitem{DBLP:conf/haskell/Finkbeiner0PS19}
\BIBentryALTinterwordspacing
B.~Finkbeiner, F.~Klein, R.~Piskac, and M.~Santolucito, ``Synthesizing
  functional reactive programs,'' in \emph{Proceedings of the 12th {ACM}
  {SIGPLAN} International Symposium on Haskell, Haskell@ICFP 2019, Berlin,
  Germany, August 18-23, 2019}, R.~A. Eisenberg, Ed.\hskip 1em plus 0.5em minus
  0.4em\relax {ACM}, 2019, pp. 162--175. [Online]. Available:
  \url{https://doi.org/10.1145/3331545.3342601}
\BIBentrySTDinterwordspacing

\bibitem{DBLP:journals/acta/FinkbeinerHLST20}
\BIBentryALTinterwordspacing
B.~Finkbeiner, C.~Hahn, P.~Lukert, M.~Stenger, and L.~Tentrup, ``Synthesis from
  hyperproperties,'' \emph{Acta Informatica}, vol.~57, no. 1-2, pp. 137--163,
  2020. [Online]. Available: \url{https://doi.org/10.1007/s00236-019-00358-2}
\BIBentrySTDinterwordspacing

\bibitem{HyperLTL}
\BIBentryALTinterwordspacing
M.~R. Clarkson, B.~Finkbeiner, M.~Koleini, K.~K. Micinski, M.~N. Rabe, and
  C.~S{\'{a}}nchez, ``Temporal logics for hyperproperties,'' in
  \emph{Principles of Security and Trust - Third International Conference,
  {POST} 2014, Held as Part of the European Joint Conferences on Theory and
  Practice of Software, {ETAPS} 2014, Grenoble, France, April 5-13, 2014,
  Proceedings}, ser. Lecture Notes in Computer Science, M.~Abadi and S.~Kremer,
  Eds., vol. 8414.\hskip 1em plus 0.5em minus 0.4em\relax Springer, 2014, pp.
  265--284. [Online]. Available:
  \url{https://doi.org/10.1007/978-3-642-54792-8\_15}
\BIBentrySTDinterwordspacing

\bibitem{conf/cav/FinkbeinerHLST18}
B.~Finkbeiner, C.~Hahn, P.~Lukert, M.~Stenger, and L.~Tentrup, ``Synthesizing
  reactive systems from hyperproperties,'' in \emph{Proceedings of {CAV}}, ser.
  LNCS, vol. 10981.\hskip 1em plus 0.5em minus 0.4em\relax Springer, 2018, pp.
  289--306.

\bibitem{DBLP:journals/corr/abs-2104-14988}
\BIBentryALTinterwordspacing
B.~Finkbeiner, P.~Heim, and N.~Passing, ``Temporal stream logic modulo
  theories,'' in \emph{Foundations of Software Science and Computation
  Structures - 25th International Conference, {FOSSACS} 2022, Held as Part of
  the European Joint Conferences on Theory and Practice of Software, {ETAPS}
  2022, Munich, Germany, April 2-7, 2022, Proceedings}, ser. Lecture Notes in
  Computer Science, P.~Bouyer and L.~Schr{\"{o}}der, Eds., vol. 13242.\hskip
  1em plus 0.5em minus 0.4em\relax Springer, 2022, pp. 325--346. [Online].
  Available: \url{https://doi.org/10.1007/978-3-030-99253-8\_17}
\BIBentrySTDinterwordspacing

\bibitem{DBLP:journals/corr/abs-2108-00090}
\BIBentryALTinterwordspacing
B.~Maderbacher and R.~Bloem, ``Reactive synthesis modulo theories using
  abstraction refinement,'' \emph{CoRR}, vol. abs/2108.00090, 2021. [Online].
  Available: \url{https://arxiv.org/abs/2108.00090}
\BIBentrySTDinterwordspacing

\bibitem{tslSygus}
\BIBentryALTinterwordspacing
W.~Choi, B.~Finkbeiner, R.~Piskac, and M.~Santolucito, ``Can reactive synthesis
  and syntax-guided synthesis be friends?'' in \emph{{PLDI} '22: 43rd {ACM}
  {SIGPLAN} International Conference on Programming Language Design and
  Implementation, San Diego, CA, USA, June 13 - 17, 2022}, R.~Jhala and
  I.~Dillig, Eds.\hskip 1em plus 0.5em minus 0.4em\relax {ACM}, 2022, pp.
  229--243. [Online]. Available: \url{https://doi.org/10.1145/3519939.3523429}
\BIBentrySTDinterwordspacing

\bibitem{DBLP:conf/fc/MavridouL18}
\BIBentryALTinterwordspacing
A.~Mavridou and A.~Laszka, ``Designing secure ethereum smart contracts: {A}
  finite state machine based approach,'' in \emph{Financial Cryptography and
  Data Security - 22nd International Conference, {FC} 2018, Nieuwpoort,
  Cura{\c{c}}ao, February 26 - March 2, 2018, Revised Selected Papers}, ser.
  Lecture Notes in Computer Science, S.~Meiklejohn and K.~Sako, Eds., vol.
  10957.\hskip 1em plus 0.5em minus 0.4em\relax Springer, 2018, pp. 523--540.
  [Online]. Available: \url{https://doi.org/10.1007/978-3-662-58387-6\_28}
\BIBentrySTDinterwordspacing

\bibitem{DBLP:conf/vstte/0001LCPDBNF19}
\BIBentryALTinterwordspacing
Y.~Wang, S.~K. Lahiri, S.~Chen, R.~Pan, I.~Dillig, C.~Born, I.~Naseer, and
  K.~Ferles, ``Formal verification of workflow policies for smart contracts in
  azure blockchain,'' in \emph{Verified Software. Theories, Tools, and
  Experiments - 11th International Conference, {VSTTE} 2019, New York City, NY,
  USA, July 13-14, 2019, Revised Selected Papers}, ser. Lecture Notes in
  Computer Science, S.~Chakraborty and J.~A. Navas, Eds., vol. 12031.\hskip 1em
  plus 0.5em minus 0.4em\relax Springer, 2019, pp. 87--106. [Online].
  Available: \url{https://doi.org/10.1007/978-3-030-41600-3\_7}
\BIBentrySTDinterwordspacing

\bibitem{DBLP:conf/atva/BonakdarpourF19}
\BIBentryALTinterwordspacing
B.~Bonakdarpour and B.~Finkbeiner, ``Program repair for hyperproperties,'' in
  \emph{Automated Technology for Verification and Analysis - 17th International
  Symposium, {ATVA} 2019, Taipei, Taiwan, October 28-31, 2019, Proceedings},
  ser. Lecture Notes in Computer Science, Y.~Chen, C.~Cheng, and J.~Esparza,
  Eds., vol. 11781.\hskip 1em plus 0.5em minus 0.4em\relax Springer, 2019, pp.
  423--441. [Online]. Available:
  \url{https://doi.org/10.1007/978-3-030-31784-3\_25}
\BIBentrySTDinterwordspacing

\bibitem{DBLP:conf/csfw/BonakdarpourF20}
\BIBentryALTinterwordspacing
------, ``Controller synthesis for hyperproperties,'' in \emph{33rd {IEEE}
  Computer Security Foundations Symposium, {CSF} 2020, Boston, MA, USA, June
  22-26, 2020}.\hskip 1em plus 0.5em minus 0.4em\relax {IEEE}, 2020, pp.
  366--379. [Online]. Available:
  \url{https://doi.org/10.1109/CSF49147.2020.00033}
\BIBentrySTDinterwordspacing

\bibitem{DBLP:conf/cav/CavadaCDGMMMRT14}
R.~Cavada, A.~Cimatti, M.~Dorigatti, A.~Griggio, A.~Mariotti, A.~Micheli,
  S.~Mover, M.~Roveri, and S.~Tonetta, ``The nuxmv symbolic model checker,'' in
  \emph{CAV}, 2014, pp. 334--342.

\bibitem{voting}
``Voting example from the solidity documentation,''
  \url{https://docs.soliditylang.org/en/v0.8.13/solidity-by-example.html\#voting},
  2022, accessed: 2022-05-12.

\bibitem{auction}
``Blind auction example from the solidity documentation,''
  \url{https://docs.soliditylang.org/en/v0.8.13/solidity-by-example.html\#blind-auction},
  2022, accessed: 2022-05-12.

\bibitem{duret.16.atva2}
A.~Duret-Lutz, A.~Lewkowicz, A.~Fauchille, T.~Michaud, E.~Renault, and L.~Xu,
  ``Spot 2.0 --- a framework for {LTL} and $\omega$-automata manipulation,'' in
  \emph{Proceedings of the 14th International Symposium on Automated Technology
  for Verification and Analysis (ATVA'16)}, ser. Lecture Notes in Computer
  Science, vol. 9938.\hskip 1em plus 0.5em minus 0.4em\relax Springer, Oct.
  2016, pp. 122--129.

\bibitem{DBLP:journals/corr/abs-1906-02906}
\BIBentryALTinterwordspacing
D.~Suvorov and V.~Ulyantsev, ``Smart contract design meets state machine
  synthesis: Case studies,'' \emph{CoRR}, vol. abs/1906.02906, 2019. [Online].
  Available: \url{http://arxiv.org/abs/1906.02906}
\BIBentrySTDinterwordspacing

\bibitem{DBLP:journals/corr/abs-1812-08829}
\BIBentryALTinterwordspacing
S.~K. Lahiri, S.~Chen, Y.~Wang, and I.~Dillig, ``Formal specification and
  verification of smart contracts for azure blockchain,'' \emph{CoRR}, vol.
  abs/1812.08829, 2018. [Online]. Available:
  \url{http://arxiv.org/abs/1812.08829}
\BIBentrySTDinterwordspacing

\bibitem{DBLP:conf/sp/PermenevDTDV20}
\BIBentryALTinterwordspacing
A.~Permenev, D.~K. Dimitrov, P.~Tsankov, D.~Drachsler{-}Cohen, and M.~T.
  Vechev, ``Verx: Safety verification of smart contracts,'' in \emph{2020
  {IEEE} Symposium on Security and Privacy, {S\&P} 2020, San Francisco, CA,
  USA, May 18-21, 2020}.\hskip 1em plus 0.5em minus 0.4em\relax {IEEE}, 2020,
  pp. 1661--1677. [Online]. Available:
  \url{https://doi.org/10.1109/SP40000.2020.00024}
\BIBentrySTDinterwordspacing

\bibitem{DBLP:conf/sp/StephensFMLD21}
\BIBentryALTinterwordspacing
J.~Stephens, K.~Ferles, B.~Mariano, S.~K. Lahiri, and I.~Dillig, ``Smartpulse:
  Automated checking of temporal properties in smart contracts,'' in \emph{42nd
  {IEEE} Symposium on Security and Privacy, {SP} 2021, San Francisco, CA, USA,
  24-27 May 2021}.\hskip 1em plus 0.5em minus 0.4em\relax {IEEE}, 2021, pp.
  555--571. [Online]. Available:
  \url{https://doi.org/10.1109/SP40001.2021.00085}
\BIBentrySTDinterwordspacing

\bibitem{DBLP:conf/csfw/ClarksonS08}
\BIBentryALTinterwordspacing
M.~R. Clarkson and F.~B. Schneider, ``Hyperproperties,'' in \emph{Proceedings
  of the 21st {IEEE} Computer Security Foundations Symposium, {CSF} 2008,
  Pittsburgh, Pennsylvania, USA, 23-25 June 2008}.\hskip 1em plus 0.5em minus
  0.4em\relax {IEEE} Computer Society, 2008, pp. 51--65. [Online]. Available:
  \url{https://doi.org/10.1109/CSF.2008.7}
\BIBentrySTDinterwordspacing

\bibitem{DBLP:conf/concur/GutsfeldMO20}
\BIBentryALTinterwordspacing
J.~O. Gutsfeld, M.~M{\"{u}}ller{-}Olm, and C.~Ohrem, ``Propositional dynamic
  logic for hyperproperties,'' in \emph{31st International Conference on
  Concurrency Theory, {CONCUR} 2020, September 1-4, 2020, Vienna, Austria
  (Virtual Conference)}, ser. LIPIcs, I.~Konnov and L.~Kov{\'{a}}cs, Eds., vol.
  171.\hskip 1em plus 0.5em minus 0.4em\relax Schloss Dagstuhl -
  Leibniz-Zentrum f{\"{u}}r Informatik, 2020, pp. 50:1--50:22. [Online].
  Available: \url{https://doi.org/10.4230/LIPIcs.CONCUR.2020.50}
\BIBentrySTDinterwordspacing

\bibitem{DBLP:journals/corr/FinkbeinerZ16}
\BIBentryALTinterwordspacing
B.~Finkbeiner and M.~Zimmermann, ``The first-order logic of hyperproperties,''
  \emph{CoRR}, vol. abs/1610.04388, 2016. [Online]. Available:
  \url{http://arxiv.org/abs/1610.04388}
\BIBentrySTDinterwordspacing

\bibitem{DBLP:conf/lics/CoenenFHH19}
\BIBentryALTinterwordspacing
N.~Coenen, B.~Finkbeiner, C.~Hahn, and J.~Hofmann, ``The hierarchy of
  hyperlogics,'' in \emph{34th Annual {ACM/IEEE} Symposium on Logic in Computer
  Science, {LICS} 2019, Vancouver, BC, Canada, June 24-27, 2019}.\hskip 1em
  plus 0.5em minus 0.4em\relax {IEEE}, 2019, pp. 1--13. [Online]. Available:
  \url{https://doi.org/10.1109/LICS.2019.8785713}
\BIBentrySTDinterwordspacing

\bibitem{DBLP:conf/cav/BaumeisterCBFS21}
\BIBentryALTinterwordspacing
J.~Baumeister, N.~Coenen, B.~Bonakdarpour, B.~Finkbeiner, and C.~S{\'{a}}nchez,
  ``A temporal logic for asynchronous hyperproperties,'' in \emph{Computer
  Aided Verification - 33rd International Conference, {CAV} 2021, Virtual
  Event, July 20-23, 2021, Proceedings, Part {I}}, ser. Lecture Notes in
  Computer Science, A.~Silva and K.~R.~M. Leino, Eds., vol. 12759.\hskip 1em
  plus 0.5em minus 0.4em\relax Springer, 2021, pp. 694--717. [Online].
  Available: \url{https://doi.org/10.1007/978-3-030-81685-8\_33}
\BIBentrySTDinterwordspacing

\bibitem{DBLP:conf/lics/BozzelliPS21}
\BIBentryALTinterwordspacing
L.~Bozzelli, A.~Peron, and C.~S{\'{a}}nchez, ``Asynchronous extensions of
  hyperltl,'' in \emph{36th Annual {ACM/IEEE} Symposium on Logic in Computer
  Science, {LICS} 2021, Rome, Italy, June 29 - July 2, 2021}.\hskip 1em plus
  0.5em minus 0.4em\relax {IEEE}, 2021, pp. 1--13. [Online]. Available:
  \url{https://doi.org/10.1109/LICS52264.2021.9470583}
\BIBentrySTDinterwordspacing

\bibitem{DBLP:journals/pacmpl/GutsfeldMO21}
\BIBentryALTinterwordspacing
J.~O. Gutsfeld, M.~M{\"{u}}ller{-}Olm, and C.~Ohrem, ``Automata and fixpoints
  for asynchronous hyperproperties,'' \emph{Proc. {ACM} Program. Lang.},
  vol.~5, no. {POPL}, pp. 1--29, 2021. [Online]. Available:
  \url{https://doi.org/10.1145/3434319}
\BIBentrySTDinterwordspacing

\bibitem{DBLP:conf/csfw/LamportS21}
\BIBentryALTinterwordspacing
L.~Lamport and F.~B. Schneider, ``Verifying hyperproperties with {TLA},'' in
  \emph{34th {IEEE} Computer Security Foundations Symposium, {CSF} 2021,
  Dubrovnik, Croatia, June 21-25, 2021}.\hskip 1em plus 0.5em minus 0.4em\relax
  {IEEE}, 2021, pp. 1--16. [Online]. Available:
  \url{https://doi.org/10.1109/CSF51468.2021.00012}
\BIBentrySTDinterwordspacing

\bibitem{DBLP:conf/post/GrishchenkoMS18}
\BIBentryALTinterwordspacing
I.~Grishchenko, M.~Maffei, and C.~Schneidewind, ``A semantic framework for the
  security analysis of ethereum smart contracts,'' in \emph{Principles of
  Security and Trust - 7th International Conference, {POST} 2018, Held as Part
  of the European Joint Conferences on Theory and Practice of Software, {ETAPS}
  2018, Thessaloniki, Greece, April 14-20, 2018, Proceedings}, ser. Lecture
  Notes in Computer Science, L.~Bauer and R.~K{\"{u}}sters, Eds., vol.
  10804.\hskip 1em plus 0.5em minus 0.4em\relax Springer, 2018, pp. 243--269.
  [Online]. Available: \url{https://doi.org/10.1007/978-3-319-89722-6\_10}
\BIBentrySTDinterwordspacing

\bibitem{DBLP:journals/pacmpl/AlbertGRRRS20}
\BIBentryALTinterwordspacing
E.~Albert, S.~Grossman, N.~Rinetzky, C.~Rodr{\'{\i}}guez{-}N{\'{u}}{\~{n}}ez,
  A.~Rubio, and M.~Sagiv, ``Taming callbacks for smart contract modularity,''
  \emph{Proc. {ACM} Program. Lang.}, vol.~4, no. {OOPSLA}, pp. 209:1--209:30,
  2020. [Online]. Available: \url{https://doi.org/10.1145/3428277}
\BIBentrySTDinterwordspacing

\end{thebibliography}
